\documentclass{amsart}
\usepackage{lineno}
\usepackage[english]{babel}
\usepackage[T1]{fontenc}
\usepackage[latin1]{inputenc}
\usepackage[all]{xy}
\usepackage{amsmath,dsfont}
\usepackage{amssymb}
\usepackage{amsthm}
\usepackage{bm}
\usepackage{mathtools}
\usepackage[margin=2.5cm]{geometry}%
\usepackage{cases}
\usepackage{upgreek}
\usepackage{braket}
\usepackage{stackengine}
\usepackage{subfigure}
\usepackage{graphicx}
\usepackage{array}
\usepackage{hyperref}
\hypersetup{
        colorlinks   = true,
        citecolor    = blue
    } 
\usepackage{float}
\usepackage{etex}
\usepackage{todonotes}

\newcommand{\RR}{\mathbb R}

\newcommand{\sdef}{\TT_{(ij)}}
\newcommand{\sgrd}{\TT_{(ij)k}}

\newcommand{\Sgrd}{\mathcal{S}\mathrm{grd}}

\newcommand{\Elas}{\mathbb{E}\mathrm{la}_{6}}
\newcommand{\Elaq}{\mathbb{E}\mathrm{la}_{4}}
\newcommand{\Elac}{\mathbb{E}\mathrm{la}_{5}}

\newcommand{\TT}{\mathbb{T}}
\newcommand{\GG}{\mathbb{G}}
\newcommand{\VV}{\mathbb{V}}
\newcommand{\WW}{\mathbb{W}}
\newcommand{\II}{\mathbb{I}}
\newcommand{\HH}{\mathbb{H}}
\newcommand{\KK}{\mathbb{K}}   

\newcommand{\ZZ}{\mathrm{Z}}
\newcommand{\DD}{\mathrm{D}}
\newcommand{\sym}{\mathbb{S}}   

\newcommand{\triv}{\mathds{1}}
\newcommand{\ode}{\mathrm{O}(2)}

\newcommand{\sod}{\mathrm{SO}(2)}

\newcommand{\GS}{\mathfrak{S}}              

\newcommand{\SymC}{\mathfrak{I}}    
  
\newcommand{\V}[1]{\underline{\mathrm{#1}}}

\newcommand{\ve}{\V{e}}                   
\newcommand{\vn}{\V{n}}   
\newcommand{\ds}{\td{s}}
\newcommand{\dt}{\td{t}}
\newcommand{\vQ}{\mathbf{Q}} 
\newcommand{\vP}{\mathbf{P}}
\newcommand{\vA}{\mathbf{A}} 
\newcommand{\vB}{\mathbf{B}} 
\newcommand{\vT}{\mathbf{T}} 
 
\newcommand{\vK}{\mathbf{K}} 
\newcommand{\vv}{\mathbf{v}} 
\newcommand{\ww}{\mathbf{w}} 
\newcommand{\vV}{\mathbf{V}} 
\newcommand{\vI}{\mathbf{I}} 
\newcommand{\vi}{\mathbf{i}}

\newcommand{\vp}{\bm{\pi}}	
\newcommand{\vR}{\mathbf{r}} 
\newcommand{\vG}{\mathbf{g}} 

\newcommand*{\trisim}{%
  \mathrel{\vcenter{\offinterlineskip
  \hbox{$\sim$}\vskip-.35ex\hbox{$\sim$}\vskip-.35ex\hbox{$\sim$}}}}

\newcommand{\sT}[1]{\underset{\trisim}{\mathrm{#1}}}
\newcommand{\cT}[1]{\underset{\approxeq}{\mathrm{#1}}}
\newcommand{\qT}[1]{\underset{\approx}{\mathrm{#1}}}
\newcommand{\tT}[1]{\underset{\simeq}{\mathrm{#1}}}
\newcommand{\dT}[1]{\underset{\sim}{\mathrm{#1}}}

\newcommand{\td}[1]{\underset{\vspace{1cm}\sim}{\mathrm{#1}}}
\newcommand{\dd}{\td{d}}

\newcommand{\otimesbar}{\; \underline{\overline{\otimes}} \;}
\newcommand{\otimesup}{\; \overline{\otimes} \;}
\newcommand{\otimesdwn}{\; \underline{\otimes} \;}

\newcommand{\id}{\dT{i}^{(2)}}
\newcommand{\Jd}{\dT{\upepsilon}}

\newtheorem{thm}{Theorem}[section]
\newtheorem{cor}[thm]{Corollary}
\newtheorem{lem}[thm]{Lemma}

\newtheorem{prop}[thm]{Proposition}
\newtheorem{rem}[thm]{Remark}
\numberwithin{equation}{section}   
\newcommand{\beq}{\begin{equation}}
\newcommand{\eeq}{\end{equation}}
\newcommand{\ben}{\begin{equation*}}
\newcommand{\een}{\end{equation*}}
\newcommand{\ba}{\begin{eqnarray}}
\newcommand{\ea}{\end{eqnarray}}
\newcommand{\ban}{\begin{eqnarray*}}
\newcommand{\ean}{\end{eqnarray*}}

\newcommand{\red}[1]{{\textcolor{black}{#1}}}

\DeclareMathOperator{\tr}{tr}

\newcommand{\udot}{\mathsurround0pt\ldotp}
\newcommand{\onedot}{$\mathsurround0pt\ldotp$}
\newcommand{\dc}{
  \mathbin{\vcenter{\baselineskip.57ex
    \hbox{\onedot}\hbox{\onedot}}%
  }}%
\newcommand{\tc}{
  \mathbin{\vcenter{\baselineskip.57ex
    \hbox{\onedot}\hbox{\onedot}\hbox{\onedot}}%
  }}%
\newcommand{\qc}{
  \mathbin{\vcenter{\baselineskip.57ex
    \hbox{\onedot}\hbox{\onedot}} \vcenter{\baselineskip.57ex
    \hbox{\onedot}\hbox{\onedot}}%
  }}%
  
\newcommand{\rdots}[1]{\overset{(#1)}{\cdot}}
\newcommand{\rayp}[1]{\overset{(#1)}{\star}}

\begin{document}


\title[Explicit harmonic structure of bidimensional \red{linear} strain-gradient elasticity]{Explicit harmonic structure of bidimensional \red{linear} strain-gradient elasticity}%

\author{N. Auffray}
\address[Nicolas Auffray]{Univ Gustave Eiffel, CNRS, MSME UMR 8208, F-77454 Marne-la-Vall\'ee, France}
\email{nicolas.auffray@univ-eiffel.fr}

\author{H. Abdoul-Anziz}
\address[Houssam Abdoul-Anziz]{Univ Gustave Eiffel, CNRS, MSME UMR 8208, F-77454 Marne-la-Vall\'ee, France}
\email{houssam.abdoulanziz@u-pem.fr}

\author{B. Desmorat}
\address[Boris Desmorat]{Sorbonne Université, CNRS, Institut Jean Le Rond d'Alembert, UMR 7190, 75005 Paris, France}
\email{boris.desmorat@sorbonne-universite.fr}

\subjclass[2010]{74B05; 74E10}
\keywords{Anisotropy; Symmetry classes; Higher-order tensors; Harmonic decomposition}%

\date{\today}%

\begin{abstract}
In the perspective of homogenization theory, strain-gradient elasticity is a strategy to describe the overall behaviour of materials with coarse mesostructure. 
In this approach, the effect of the mesostructure is described by the use of three elasticity tensors whose orders vary from 4 to 6. 
Higher-order constitutive tensors make it possible to describe rich physical phenomena. However, these objects have intricate algebraic structures that prevent us from having a clear picture of their modeling capabilities.
The harmonic decomposition is a fundamental tool to investigate the anisotropic properties of constitutive tensor spaces. For higher-order tensors (i.e. tensors of order $n\geq$3), \red{their determination} is generally a difficult task. In this paper, a novel procedure to obtain this decomposition is introduced. This method, which we have called the \textit{Clebsch-Gordan Harmonic Algorithm}, allows one to obtain \emph{explicit} harmonic decompositions satisfying  good properties such as orthogonality and uniqueness. The elements of the decomposition also have a precise geometrical meaning simplifying their physical interpretation. This new algorithm is here developed in the specific case of 2D space and applied to Mindlin's Strain-Gradient Elasticity. We provide, for the first time, the harmonic decompositions of the fifth- and sixth-order elasticity tensors involved in this constitutive law. The \textit{Clebsch-Gordan Harmonic Algorithm} is not restricted to strain-gradient elasticity and may find interesting applications in different fields of mechanics which involve higher-order tensors.
\end{abstract}

\maketitle





 
\section*{Introduction}

\subsection*{Strain-Gradient Elasticity}

Continuum mechanics is a well-established theory which constitutes the classical framework to study strain and stress in solid materials. 
The physics contained in the theory is versatile enough to describe the inner state of a planet subjected to the force of gravity, as well as to meet the daily needs of a mechanical engineer. These successes make classical continuum mechanics a fundamental theory of modern physics. Nevertheless, despite all its successes, situations have arisen in which its classical formulation reaches its limits and fails to correctly describe the physics at work: mechanics of nano-structures \cite{Zhu10,Ere16,CFB16}, elastic waves in periodic continua \cite{Vin86,RA15}, capillarity and surface tension phenomena \cite{Cas61,Min65,Sep89,For20}, etc. Despite their diversities, these examples have in common that they show dependencies to characteristic lengths, a property that cannot be taken into account within the classical formulation of continuum mechanics.\\

Since the pioneering work of the Cosserat brothers in the early years of the 20th century \cite{Cos09}, many scientists have proposed enriched continuum theories to extend the capabilities of the standard theory. With the contributions of Koiter \cite{Koi64}, Toupin \cite{Tou62}, Eringen \cite{Eri68}, Mindlin \cite{Min64,Min65,ME68} and many others, the 60' was probably the most fertile period of this project \cite{Mau11}. At that time, the community of theoretical mechanics developed many models that are still relevant today. Leaving aside non-local aspects, the classical theory can be enriched either by including additional degrees of freedom (DOF) or  by including higher-gradients of the DOF in the mechanical energy. All these models suffer the drawbacks: 1) to require too many material parameters to be of practical uses and 2) due to the intrinsic complexity of the equations, the content of the physics described by these models is difficult to grasp. At that time, except for some applications in physics to describe the dispersion of elastic waves in crystals \cite{Vin86}, or to describe the mechanics of liquid crystals \cite{Eri78}, these models were mostly unused.\\

The development of numerical homogenization methods that allow the coefficients required by the extended models to be related to a known mesostructure has changed the situation \cite{For98,BG10,TJA+12}. There is now a renewed interest in generalized continuum theories \cite{VLM07,AAE10,BG14b,BF20}.  This renewed interest has also been supported by the emergence of additive manufacturing techniques that permit specific mesostructures to be fabricated (almost) on demand \cite{WCC+15,ABA+17,HLZ17}. Recently theoretical approaches have been developed to design architectured materials in order to maximise the non-standard effects predicted by generalized mechanics \cite{BIH+20}, and for the first time the associated materials have been produced and tested \cite{PSM+18}.\\

In statics, to have an important contribution in the overall behaviour of strain-gradient effects, the classical elasticity needs to be almost degenerated \cite{AS18,JS20}. 
 This aspect has been widely studied in pantographic structures \cite{ASI03,ISS+19}, and is also encountered in pentamode metamaterials \cite{MC95,KBS12} which also possess quasi-degenerated deformation modes.  
In dynamics, the contribution of  higher-order effects is easier to highlight and to control \cite{AA11,BG14b,BCG18}. In some recent contributions \cite{RPA18,RA19}, it has been shown that for honeycomb materials, higher-order effects make a significant contribution as soon as the wavelength is about 10 times the size of unit cell size. This effect has been used in \cite{RA19} to bend the trajectory of an elastic wave around a circular hole. The control of this effect in more general situations can find interesting engineering applications \cite{BUA+17,RRA18}.\\

The design of architectured materials in which these higher-order effects are maximised, or at the contrary inhibited, is at the present time an open and challenging problem in mechanical sciences. An approach to the optimal design of strain-gradient architectured materials is to use topological optimisation algorithms \cite{DYL+18,RIB+19,AGP19,DA19}. For isotropic continua the material optimisation process is formulated  by  expressing the design functionals  as functions of the isotropic parameters of the different constitutive tensors\cite{BS13}. For anisotropic materials the previous approach cannot be extended without a few precautions. A path to this rigorous extension is through the use of tensor invariants \cite{RIB+19,Ran19}.

\subsection*{Harmonic decomposition}

To optimize the material independently of its spatial orientation, design functionals should be expressed as functions of tensor invariants rather than in terms of tensor components. The associated mathematical theory is the \emph{invariant theory}, which aims at determining the minimal number of tensorial functions that are invariant with respect to a given group of transformations. Here, since a material is left invariant by orthogonal transformations (rotations and \red{improper transformations}), the group of transformations will be $\mathrm{O}(d)$  in a $d$-dimensional space. If the mathematical theory is clear in any dimension, its practical and effective application strongly depend on the space dimension. For $d=2$ the situation is rather clear and general results are available \cite{Zhe94,DOA+20}, while the case  $d=3$ presents some serious difficulties preventing general results \cite{BKO94,SB97,OKA17,CCQ+19}.

In both cases, the effective construction of a basis of invariants lies on a first step which is the determination of an explicit irreducible decomposition of the constitutive tensors.
In the case of a symmetric second-order tensor, this decomposition corresponds to the decomposition of a tensor into a spherical and a deviatoric part. This approach can be generalized to tensors of arbitrary order \cite{JCB78,Zhe94,AKO16}. If the formal structure of these decompositions is easy to determine, both in 2D and 3D, when it comes to determining an explicit decomposition formula, things become more difficult:
\begin{itemize}
\item even if the harmonic structure is uniquely defined, the explicit decomposition is, in general, not unique and hence different, albeit isomorphic, constructions are possible \cite{GSS89};
\item without a pertinent choice for the explicit harmonic decomposition, the derived tensor invariants will not have a  clear physical content;
\item the complexity of the computations increases quickly with the tensor order. An algorithmic procedure is, therefore, mandatory to perform this decomposition. 
\end{itemize} 

In the present study, which is devoted to bidimensional strain-gradient elasticity, the constitutive law involves tensors of order ranging from 4 to 6.  If the harmonic decomposition of the fourth-order elasticity tensor is well known \cite{Via97,DA19}, for the fifth- and sixth-order tensors no decompositions are, to the best of the authors' knowledge,  available in the literature. Some procedures to perform this decomposition are available in the literature of continuum mechanics:
\begin{itemize}
\item Spencer's Algorithm \cite{Spe70}:  Spencer's method  consists in first reducing a general tensor into totally symmetric tensors and then decomposing each totally symmetric tensor into totally symmetric traceless tensors. It is, at the present time, the most known and used algorithm \cite{FV96,FV97,LQH11}. However, this approach suffers the following limitations:
\begin{itemize}
\item the treatment of higher-order tensors is quickly intractable;
\item its numerical implementation is not straightforward.
\end{itemize}
If efficient in simple situations, the Spencer's algorithm seems to be of limited interest to treat more complicated problems.
\item Verchery's Method \cite{Ver82,Van07}:  Verchery's approach is based on a complex change of variable, in a sense analogous to the transformation of
Green and Zerna \cite{GZ54}, \red{which can also be rooted to the work of Kolosov and Muskhelishvili \cite{Mus13}}. This method, which was recently re-explored in \cite{DOA+20}, is elegant, but its main limitations are:
\begin{itemize}
\item it does not produce a practical formula for the decomposition;
\item when different harmonic terms of the same order are present, the pairing of their components is not direct;
\item the method is restricted to the 2D cases.
\end{itemize}
\item Zou's Approach \cite{ZZ00,ZZD+01}: The Zou's method exploits a \emph{Clebsch-Gordan} identity to construct an orthogonal harmonic decomposition of a $n$th-order tensor from the  orthogonal harmonic decomposition of a $(n-1)$th-order tensor. This iterative method is powerful to obtain orthogonal harmonic decomposition of higher-order tensors without index symmetry. Its application to higher-order tensors having specific index symmetries is possible, but can be cumbersome.
\end{itemize}
In addition,  as a common limitation, none of the three listed methods provides a clear mechanical content to the harmonic tensors of the resulting decomposition.\\

The objective of the present contribution is to introduce a new method for determining explicit harmonic decompositions \red{that solves the problems identified with the previous methods}. The algorithm we propose, which will be referred to as the \emph{Clebsch-Gordan Harmonic Algorithm}, will be conducted here in a 2D framework, but it can be extended without any conceptual obstruction to the 3D framework.

The main idea of the Clebsch-Gordan approach, and the main difference from all other methods, is first to decompose not the constitutive tensor (e.g.  the fourth-order elasticity tensor for classical elasticity) but the state tensors on which it acts (e.g. the strain/stress second-order tensor). This first decomposition will induce a \emph{block structure} on the constitutive tensor. If the  \emph{elementary blocks} are generally not harmonic, their harmonic structures are very simple, and their decomposition into irreducible tensors  easy to proceed. The combination of these different steps leads to an explicit harmonic decomposition of the constitutive tensor. \red{An example of this situation is provided in Appendix \ref{sec:TypeII}.}
\red{Compared to the other methods}, the main advantages of the proposed approach are:
\begin{enumerate}
\item the procedure is algorithmic \red{and can easily be implemented in a Computer Algebra System};
\item \red{once a decomposition of the state tensor spaces chosen, the resulting decomposition of the constitutive tensors is uniquely defined};
\item the elements of the decomposition are orthogonal to each other and have a clear physical content.
\end{enumerate}
Further, since the space of state tensors is first decomposed, the resulting harmonic decomposition implies a decomposition of the internal energy density. Such a property is valuable to provide a physical content to the higher-order constitutive parameters of the model.\\

\red{Before closing this introduction, it is important to mention that the proposed method is introduced here  in the context of strain-gradient elasticity only because of the authors' interest in this application.  Tensors of orders greater than two are necessary in many areas of mechanics.  They appear, for example, in generalised continuum theories such as in strain-gradient elasticity \cite{Min64} or micromorphic elasticity \cite{Eri68}, but also in higher-order models of the plastic yielding function \cite{Got77}, for the study of the acoustoelastic effect \cite{TB61}, in continuum damage mechanics in  which fabric tensors are involved \cite{Kan84} and, of course, for elastic shells \cite{EP06}. We believe that our approach to harmonic decomposition can find interesting applications in these different fields.}\\

\textbf{This article is organized as follows.} In Section 1, we introduce notations that we will use throughout the text. Section 2 is devoted to the description of the strain-gradient elasticity constitutive law. By the end of this section, the harmonic structure of the model is  introduced and detailed. Section 3 is devoted  to the theoretical aspects of the method. Some general results that will be used all along the paper are introduced. In Section 4 the method is  detailed for the fourth-order elasticity tensor. The purpose of this section is mainly illustrative and aims a recovering the well-known harmonic decomposition of the elasticity tensor. This approach is then extended, in Section 5, to the decomposition of the fifth- and sixth-order elasticity tensors involved in the model. Those results are, we believe, new and not available in the literature.

\section{Notations}

Throughout this paper, the two-dimensional Euclidean space $\mathcal{E}^2$ is equipped with a rectangular Cartesian coordinate system with origin  $\mathrm{O}$ and an orthonormal basis $\mathcal{B}=\{\ve_{1},\ve_{2}\}$. Upon the choice of a reference point $\mathrm{O}$ in $\mathcal{E}^2$ and a given basis $\mathcal{B}$, $\mathcal{E}^2$ will be identified with the $2$-dimensional vector space $\RR^2$. As a consequence, points in $\mathcal{E}^2$  will be designated by their vector positions with respect to $\mathrm{O}$. In the following, $\V{x} = (x_1, x_2) = x_1 \ve_1 + x_2 \ve_2 = x_i \ve_i$, where Einstein summation convention is used, i.e. when an index appears twice in an expression it implies summation of this term over all the values of the index. Below are provided some specific notations and conventions used in this article.

\subsubsection*{Groups:}
\begin{itemize}
  \item $\ode$: the group of invertible transformations of $\RR^2$ satisfying $\vG^{-1} = \vG^{T}$, where $\vG^{-1}$ and $\vG^{T}$ stand for the inverse and the transpose of $\vG$. This group is called the orthogonal group;
    \item $\sod$: the subgroup of $\ode$ of elements with determinant $1$, called the special orthogonal group;
    \item $\GS_{n}$: the group of all permutations on the set $\{1,2,...,n\}$, called the symmetric group;
   \item $\triv$: the  trivial group solely containing the identity. 
\end{itemize}

As a matrix group, $\ode$ is generated by
        \begin{equation*}
          \vR(\theta) =
          \begin{pmatrix}
            \cos \theta & -\sin \theta \\
            \sin \theta & \cos \theta
          \end{pmatrix} \quad \text{with $0 \leq \theta < 2\pi$} \quad \text{and} \quad \vp(\ve_2) =
          \begin{pmatrix}
            1 & 0  \\
            0 & -1
          \end{pmatrix},
        \end{equation*}
where $\vR(\theta) $ is the rotation by an angle $\theta$ and $\vp(\vn)$ is the reflection across the line normal to $\vn$:
\ben
\vp(\vn):= \dT{i}^{(2)}-2\vn\otimes\vn,\quad \|\vn\|=1,
\een
with $\dT{i}^{(2)}$  the second order identity tensor, as defined below.

\red{To describe the symmetry classes of tensors, the following finite subgroups of $\ode$ will be used:
\begin{itemize}
\item $\ZZ_{k} (k \geq 2)$ the cyclic group with $k$ elements, generated by $\vR(2\pi /k)$. For $k=1$, we have $\ZZ_{k}=\triv$;
\item $\DD^{\vn}_{k} (k \geq 2)$ the dihedral group  with $2k$ elements, generated by $\vR(2\pi /k)$ and $\vp(\vn)$. For  $k=1$, the group $\DD^{\vn}_{1}$ will be denoted by $\ZZ^{\pi,\vn}_{2}$. 
\end{itemize}
It has to be noted that, up to conjugacy by an element of $\sod$, any $\DD^{\vn}_{k}$ is conjugate to $\DD^{\ve_{2}}_{k}$, and will simply be denoted by $\DD_{k}$.  Accordingly, $\ZZ^{\pi,\ve_{2}}_{2}$ will simply be denoted $\ZZ^{\pi}_{2}$. }

\subsubsection*{Tensor products:}
\begin{itemize}
  \item $\otimes$ denotes the usual tensor product of two tensors or vector spaces; 
  \item $\overset{n}{\otimes}$ denotes the $n$th power of the tensor product, for example $\overset{n}{\otimes} V: = V \otimes \cdots \otimes V$ ($n$ copies of $V$);
  \item $\otimes^s$ denotes the symmetrized tensor product;
   \item $\otimesup$ and $\otimesdwn$ indicate the following partial twisted tensor products \red{between second-order tensors:}
\begin{equation*}
\left( \dT{a} \otimesup \dT{b} \right)_{ijkl} = a_{ik}b_{jl},\ \left( \dT{a} \otimesdwn \dT{b} \right)_{ijkl} = a_{il}b_{jk};
\end{equation*}
  \item $\otimesbar$ indicates the twisted tensor product defined by
\begin{equation*}
\left( \dT{a} \otimesbar \dT{b} \right)= \frac{1}{2}\left( \dT{a} \otimesup \dT{b}+ \dT{a} \otimesdwn \dT{b}\right).
\end{equation*}
\end{itemize}

\subsubsection*{Tensor spaces:}
\begin{itemize}
 \item  $\GG^n: = \overset{n}{\otimes} \RR^2$ is the space of $n$th-order tensors having no index symmetries;
 \item  $\TT^n$ is a subspace of  $\GG^n$ defined by its index symmetries;
 \item $\sym^n$ is the space of \emph{totally symmetric}\footnote{By totally symmetric we mean symmetric with respect to all permutations of indices.} $n$th-order tensors on $\RR^2$;
 \item $\KK^{n}$ is the space of $n$th-order \emph{harmonic tensors} (i.e. totally symmetric and traceless tensors), with 
\[
\dim(\KK^n)=
\begin{cases}
2 & \text{if $n > 0$}, \\
1 & \text{if $n\in\{0, -1\}$}.
\end{cases}
\]
Among them: 
\begin{itemize}
\item $\KK^0$ is the space of scalars;
\item $\KK^{-1}$ is the space of \emph{pseudo-scalars} (i.e. scalars which change sign under improper transformations).
\end{itemize}
\end{itemize}
We will use tensors of different orders, tensors of order $-1$, $0$, $1$, $2$, $3$, $4$, $5$ and $6$ are denoted by $\beta$, $\alpha$, $\V{v}$, $\dT{a}$, $\tT{A}$, $\qT{B}$, $\cT{C}$, $\sT{D}$, respectively. General tensors (i.e. with no mention of their order) are denoted using bold fonts, as for instance $\vT$.
With respect to $\mathcal{B}$, the components of $\vT\in\TT^n$ are denoted as
\ben
\mathbf{T}=T_{i_{1}\ldots i_{n}}.
\een
The simple, double, triple and fourth-order contractions are written $\udot$, $\dc$, $\tc$, $\qc$, respectively. Generic $k$th-order contraction will be indicated by the notation $\rdots{k}$. In components with respect to $\mathcal{B}$, for general tensors $\vA$ and $\vB$, these notations correspond to
\begin{equation*}
(\vA \udot \vB)_{i_1 \ldots i_n} = A_{i_1 \ldots i_p j} B_{j i_{p+1} \ldots i_n}, \qquad (\vA \dc \vB)_{i_1 \ldots i_n} = A_{i_1 \ldots i_p jk} B_{jk i_{p+1} \ldots i_n},
\end{equation*}
\begin{equation*}
(\vA \tc \vB)_{i_1 \ldots i_n} = A_{i_1 \ldots i_p jkl} B_{jkl i_{p+1} \ldots i_n}, \qquad (\vA \qc \vB)_{i_1 \ldots i_n} = A_{i_1 \ldots i_p jklm} B_{jklm i_{p+1} \ldots i_n}.
\end{equation*}
When needed, index symmetries of both spaces and their elements are expressed as follows: $(..)$ indicates invariance under permutations of the indices in parentheses and $\underline{..}\ \underline{..}$ indicates symmetry with respect to permutations of the underlined blocks. For example, $\dT{a} \in \TT_{(ij)}$ means that $a_{ij} = a_{ji}$.

\subsubsection*{Actions on tensors}
We consider the action of two groups on the space $\TT^n$:
\begin{itemize}
\item  the action of $\ode$,  given by 
\ban
\forall g\in\ode,\  \mathbf{T}^g:=g\overset{n}{\star}\mathbf{T}=g_{i_{1}j_{1}}g_{i_{2}j_{2}}\ldots g_{i_{n}j_{n}}T_{j_{1}j_{2}\ldots j_{n}}\ve_{i_{1}}\otimes\ldots\otimes\ve_{i_{n}}.
\ean
This action is the tensorial action, sometimes also known as the Rayleigh product. When clear from the context,  no mention will be made of the tensor order in the product, in this case the notation $\overset{n}{\star}$ simplifies to $\star$. The set $
\mathrm{G}(\mathbf{T}):=\{g\in\ode\mid g\star\mathbf{T}=\mathbf{T}\}$
is called the \emph{spatial symmetry group} of $\mathbf{T}$. A tensor $\mathbf{T}$ is said to be \emph{isotropic} if $\mathrm{G}(\mathbf{T})=\ode$;
\item  the action of $\GS_{n}$, given by 
\ben
\forall \varsigma\in\GS_{n},\ \mathbf{T}^\varsigma:=\varsigma \ast\mathbf{T}=T_{i_{\varsigma(1)}i_{\varsigma(2)}\ldots i_{\varsigma(n)}}\ve_{i_{1}}\otimes\ldots\otimes\ve_{i_{n}}.
\een
The set $\mathcal{G}(\mathbf{T}):=\{\varsigma\in\GS_{n}\mid \varsigma \ast\mathbf{T}=\mathbf{T}\}$ is called the \emph{index symmetry group} of $\mathbf{T}$. A tensor $\mathbf{T}$ is said to be
\begin{itemize}
\item \emph{generic} if $\mathcal{G}(\mathbf{T})=\triv$, elements of $\GG^n$ verify this property;
\item\emph{totally symmetric} if $\mathcal{G}(\mathbf{T})=\GS_{n}$, elements of $\sym^n$ verify this property.
\end{itemize}
The notation $\mathcal{G}(\mathbb{T}^n)$ will also be used to indicate the \emph{index symmetry group} of a generic element $\mathbf{T}\in\mathbb{T}^n$.
\end{itemize}

\subsubsection*{Special tensors:}
\begin{itemize}
\item $\vI^{\mathbb{V}}$ is the identity tensor on the vector space $\mathbb{V}$. Identity tensors can be expressed using isotropic tensors \cite{Wey46} :
\begin{itemize}
\item $\dT{I}^{\RR^2}$ is the second-order identity tensor on  $\RR^2$. It is defined from $\dT{i}^{(2)}$ which components are given by the Kronecker delta $\delta_{ij}$:
\ben
\dT{I}^{\RR^2}=\dT{i}^{(2)};
\een
\item $\qT{I}^{\sdef} = \dT{i}^{(2)} \otimesbar \dT{i}^{(2)}$ is the fourth-order identity tensor on $\sym^{2}$.
\end{itemize}
\item $\Jd$ denotes the 2D Levi-Civita tensor defined by
\begin{equation*}
\upepsilon_{ij} = 
\begin{cases}
1 & \text{if $ij = 12$}, \\
- 1 & \text{if $ij = 21$}, \\
0 & \text{if $i=j$}.
\end{cases}
\end{equation*}
\end{itemize}

\subsubsection*{Miscellaneous notations:}
\begin{itemize}
\item $\simeq$ denotes  an isomorphism;
\item  $\mathcal{L}(\mathbb{E},\mathbb{F})$ indicates the space of linear maps from $\mathbb{E}$ to $\mathbb{F}$; 
\item  $\mathcal{L}(\mathbb{E})$ indicates the space of linear maps from $\mathbb{E}$ to $\mathbb{E}$; 
\item  $\mathcal{L}^s(\mathbb{E})$ indicates the space of self-adjoint linear maps on $\mathbb{E}$.
\end{itemize}

\subsubsection*{Tensor isotropic basis}\label{s:IsoSys}

Let us introduce $\II^{n}$ the space of $n$th-order isotropic tensors:
 \ben
 \II^{n}:=\{\vT\in\TT^n|\quad \forall g\in\ode,\quad g\star\vT=\vT \}.
 \een
Elements of  $\II^{n}$ are denoted $\vi^{(n)}_{p}$, in which $n$ indicates the order of the tensor and $p$ distinguishes among the different isotropic tensors of the same order.
Every isotropic tensor can be expressed as a linear combination of products of $\dT{i}^{(2)}$ \cite{Wey46}. Products of $\dT{i}^{(2)}$ will be referred to as elementary isotropic tensors and, by definition, the element $\vi^{(2p)}_{1}$ is defined as:
\ben
\vi^{(2p)}_{1}=\dT{i}^{(2)}\otimes\ldots\otimes\dT{i}^{(2)}=\dT{i}^{(2)}\overset{p-1}{\otimes}\dT{i}^{(2)}.
\een
For fourth-order tensors, there exists $3$ elementary isotropic tensors\footnote{with a slight abuse of notation.}:
\beq\label{eq:TIso4}
\qT{i}^{(4)}_1=\delta_{ij}\delta_{kl},\ \qT{i}^{(4)}_2=\delta_{ik}\delta_{jl},\ \qT{i}^{(4)}_3=\delta_{il}\delta_{jk}.
\eeq
\red{For latter use, this allows us to express the identity tensor on second-order symmetric tensors as
\beq\label{eq:IsoS2}
 \qT{I}^{\sdef}=\frac{1}{2}\left(\qT{i}^{(4)}_2+\qT{i}^{(4)}_3\right)
\eeq}
For sixth-order tensors, there exists $15$ elementary isotropic tensors:
\ban\label{eq:TIso6}
\sT{i}^{(6)}_1=\delta_{ij}\delta_{kl}\delta_{mn},&\sT{i}^{(6)}_2=\delta_{ij}\delta_{km}\delta_{ln},&\sT{i}^{(6)}_3=\delta_{ij}\delta_{kn}\delta_{lm}\\
\sT{i}^{(6)}_4=\delta_{ik}\delta_{jl}\delta_{mn},&\sT{i}^{(6)}_5=\delta_{ik}\delta_{jm}\delta_{ln},&\sT{i}^{(6)}_6=\delta_{ik}\delta_{jn}\delta_{lm}\\
\sT{i}^{(6)}_7=\delta_{il}\delta_{jk}\delta_{mn},&\sT{i}^{(6)}_8=\delta_{il}\delta_{jm}\delta_{kn},&\sT{i}^{(6)}_9=\delta_{il}\delta_{jn}\delta_{km}\\
\sT{i}^{(6)}_{10}=\delta_{im}\delta_{jk}\delta_{ln},&\sT{i}^{(6)}_{11}=\delta_{im}\delta_{jn}\delta_{kl},&\sT{i}^{(6)}_{12}=\delta_{im}\delta_{jl}\delta_{kn}\\
\sT{i}^{(6)}_{13}=\delta_{in}\delta_{jk}\delta_{lm},&\sT{i}^{(6)}_{14}=\delta_{in}\delta_{jl}\delta_{km},&\sT{i}^{(6)}_{15}=\delta_{in}\delta_{jm}\delta_{kl}.
\ean
According to the dimension of the physical spaces, these elementary tensors may not be necessarily independent. According to Racah  \cite{Rac33}, in 2D the number of independent fourth-order isotropic tensors is still 3, while for sixth-order isotropic tensors only 10 are independent.


\section{Strain-gradient \red{linear} elasticity} \label{sec:law-flexo}
We introduce in this section the constitutive law of a linear strain-gradient elastic material \cite{Min64,ME68}. First, we present the state and constitutive tensors of the model. And, in the second part of the section, we detail their harmonic structures,  a step mandatory for constructing the harmonic decomposition.

\subsection{Constitutive equations}

\emph{State tensors} describe point-wisely the different physical fields (primal and dual) of the model. A linear constitutive law can be viewed as a linear map between the state tensors that characterize a chosen physical model. It is defined by a set of \emph{constitutive tensors}  which describe  the influence of the matter on these state tensor fields. More precisely, they describe how primal and dual fields are connected by the matter \cite{Ton13}.

In the case of \emph{classical elasticity}, the state tensors are $\dT{\upsigma}$ and $\dT{\upvarepsilon}$. These tensors characterize the local state of stress and of strain, respectively\footnote{In the infinitesimal setting, the strain tensor is defined from the displacement field $\V{u}$ as $\dT{\upvarepsilon}: = \frac{1}{2} (\V{u} \otimes \V{\nabla} + \V{\nabla} \otimes \V{u})$, where $\V{\nabla}$ denotes the nabla differential operator.}, and belong to the same space $\sdef$. The linearity of the model implies the use of a 
fourth-order tensor $\qT{C}$  as a constitutive tensor. This tensor can be viewed as an element of $\mathcal{L}^s(\sdef,\sdef)$.
In summary, for classical elasticity:
\begin{itemize}
\item	 \emph{State tensors}: $\dT{\upsigma}$, $\dT{\upvarepsilon}$;
\item	 \emph{Constitutive tensor}: $\qT{C}$.
\end{itemize}

The linear strain-gradient elasticity model \cite{Min65,ME68} is obtained by extending the set of state tensors by including the strain-gradient tensor $\tT{\upeta}: = \dT{\upvarepsilon} \otimes \V{\nabla}$ and its dual quantity, the hyperstress tensor $\tT{\uptau}$. \red{This construction corresponds to the Mindlin's \emph{Type II} formulation \cite{ME68}}. Those tensors are elements of $\sgrd$.  The constitutive equations of the model define the stress tensor $\dT{\upsigma}$ and the hyperstress tensor $\tT{\uptau}$ as linear functions of the strain tensor $\dT{\upvarepsilon}$ and the strain-gradient tensor $\tT{\upeta}$. This coupled constitutive law requires tensors belonging to the following spaces:
\ben
\qT{C}\in \Elaq\simeq\mathcal{L}^s(\sdef,\sdef),\quad\cT{M}\in \Elac\simeq\mathcal{L}(\sgrd,\sdef),\quad \sT{A}\in \Elas\simeq\mathcal{L}^s(\sgrd,\sgrd).
\een
In this model we have:
\begin{itemize}
\item	 \emph{State tensors}: $\dT{\upsigma}$, $\dT{\upvarepsilon}$, $\tT{\uptau}$, $\tT{\upeta}$;
\item	 \emph{Constitutive tensors}: $\qT{C}$, $\cT{M}$ and $\sT{A}$.
\end{itemize}
To be more specific, the constitutive equations read:
\begin{equation}
\begin{cases}
\dT{\upsigma} & = \qT{C} \dc \dT{\upvarepsilon} + \cT{M} \tc \tT{\upeta} \\ 
\tT{\uptau} & = \cT{M}^{\top} \dc \dT{\upvarepsilon} + \sT{A} \tc \tT{\upeta}
\end{cases}
\end{equation}
where 
\begin{itemize}
     \item $\qT{C} \in \Elaq:= \Set{\qT{T} \in \overset{4}{\otimes} \RR^2 \mid \qT{T} \in \TT_{\underline{(ij)}\ \underline{(kl)}}}$ is the fourth-order elasticity tensor;
   \item $\cT{M} \in \Elac:= \Set{\cT{T} \in \overset{5}{\otimes} \RR^2 \mid \cT{T} \in \TT_{(ij)(kl)m}}$ is the fifth-order elasticity tensor;
   \item $\cT{M}^{\top} \in \Elac^{\top}:= \Set{\cT{T} \in \overset{5}{\otimes} \RR^2 \mid \cT{T} \in \TT_{(ij)k(lm)}}$ is the fifth-order elasticity tensor defined as the transpose of $\cT{M}$ in the following sense $(\cT{M}^\top)_{ijklm} = M_{lmijk}$;
   \item $\sT{A} \in \Elas:= \Set{\sT{T} \in \overset{6}{\otimes} \RR^2 \mid \sT{T} \in \TT_{\underline{(ij)k}\ \underline{(lm)n}}}$ is the sixth-order elasticity tensor.
\end{itemize}
Let's define $\Sgrd$ the tensor space of the strain-gradient constitutive tensors as
 \begin{equation} \label{esp-Flex}
\Sgrd = \ \Elaq \oplus \Elac \oplus \Elas.
\end{equation}
A strain-gradient elastic law is defined by a triplet $\mathcal{E}:=\Big(\qT{C}, \cT{M}, \sT{A} \Big)\in \Sgrd$ .


\subsection{Harmonic structure of constitutive tensors}

When a material is rotated\footnote{Here \emph{rotated} is understood in the broad sense of a full orthogonal transformation.} its physical nature is not affected but, with respect to a fixed reference, \red{its} constitutive tensors are transformed. Since constitutive tensors are usually of order greater than 2,  the way they transform is not simple and their different parts  transform differently: some components are left fixed while others \emph{turn} at different speeds. 
The different mechanisms of transformation of a tensor with respect to an orthogonal transformation are revealed by its harmonic structure\footnote{The explicit harmonic decomposition is just an explicit expression of  this harmonic structure.}. The harmonic decomposition consists in decomposing a finite-dimensional vector space into a direct sum of \emph{$\ode$-irreducible subspaces}. A subspace $\mathbb{K}$ of $\TT^n$ is called $\ode$-irreducible if: 1) it is \emph{$\ode$-invariant} (i.e., $\vG \star \mathrm{T} \in \mathbb{K}$ for all $\vG \in \ode$ and $\mathrm{T} \in \mathbb{K}$); 2) its only invariant subspaces are itself and the null space. It is known that $\ode$-reducible spaces are isomorphic to a direct sum  of harmonic tensor spaces  $\KK^n$ \cite{GSS89,AKO16}.
Such a decomposition is interesting since the $\ode$-action on $\KK^n$ is elementary and given by $\rho_n$ \cite{AKO16}, with for $n \geq 1$:
\begin{equation}\label{eq:TrsIrr}
\rho_n (\vR(\theta)): = 
\begin{pmatrix}
\cos n \theta & - \sin n \theta \\
\sin n \theta & \cos n \theta
\end{pmatrix}, 
\qquad 
\rho_n (\vp(\ve_2)):=
\begin{pmatrix}
1 & 0 \\
0 & -1
\end{pmatrix}.
\end{equation} 
The $\ode$-action on $\KK^0$ is the identity and the $\ode$-action on $\KK^{-1}$ is given by the determinant of the transformation:
\begin{equation}\label{eq:TrsIrr0}
\rho_0 (\vQ): = 1,
\qquad 
\rho_{-1} (\vQ):=\det \vQ.
\end{equation}
The harmonic structure of a tensor space can be determined without making heavy computations by using the \emph{Clebsch-Gordan formula}. This formula indicates how the tensor product of two irreducible spaces decomposes into a direct sum of irreducible spaces. Note that this formula only indicates the structure of the resulting vector space and does not provide an explicit construction of the decomposition.  The construction of an associated explicit decomposition will be undertaken in Sections 4 and 5.
For the determination of the harmonic structure, we use the following result, the proof of which is found in \cite{AKO16}.
\begin{lem} \label{lem-Gordan}
For every integers $p > 0$ and $q > 0$, we have the following isotypic decompositions, where the meaningless products are indicated by $\times$:
\begin{table}[H]
\begin{tabular}{|c|c|c|c|}
\hline 
$\otimes$ & $\KK^{q}$ & $\KK^0$ & $\KK^{-1}$ \\
\hline
$\KK^p$ & $\begin{cases}
\KK^{p+q} \oplus \KK^{\lvert p-q \rvert}, & \text{$p \neq q$} \\
\KK^{2p} \oplus \KK^0 \oplus \KK^{-1}, & \text{$p=q$}
\end{cases}$
& $\KK^p$ & $\KK^p$ \\
\hline
$\KK^0$ & $\KK^q$ & $\KK^0$ & $\KK^{-1}$ \\
\hline
$\KK^{-1}$ & $\KK^q$ & $\KK^{-1}$ & $\KK^0$ \\
\hline
\end{tabular} \hfill
\begin{tabular}{|c|c|c|c|}
\hline 
$\otimes^{s}$ & $\KK^{p}$ & $\KK^0$ & $\KK^{-1}$ \\
\hline
$\KK^p$ &  $\KK^{2p} \oplus \KK^0$ & $\times$ & $\times$ \\
\hline
$\KK^0$ & $\times$ & $\KK^0$ & $\times$ \\
\hline
$\KK^{-1}$ & $\times$ & $\times$ & $\KK^0$ \\
\hline
\end{tabular}
\end{table}
\end{lem}
Using the previous result, we can determine the harmonic structure of the  state tensor spaces and constitutive tensor spaces of the strain-gradient elasticity:\\
\begin{table}[H]\label{tab:IrrDec}
    \renewcommand{\arraystretch}{1.1}
    \begin{tabular}{|c||c|}
      \hline
       State tensor space & Harmonic structure \\ \hline  \hline
      $\sdef$   &  $\KK^2\oplus\KK^0$\\  \hline
      $\sgrd$   &$\KK^3\oplus2\KK^1$\\  \hline
    \end{tabular},
\hfill
   \renewcommand{\arraystretch}{1.1}
    \begin{tabular}{|c||c|}
      \hline
      Constitutive tensor space & Harmonic structure \\ \hline  \hline
$\Elaq$&$\KK^4\oplus\KK^2\oplus 2 \KK^0$\\ \hline
$\Elac$ &$\KK^5\oplus 3 \KK^3\oplus 5 \KK^1$\\ \hline
$\Elas$ &$\KK^6\oplus2\KK^4\oplus5\KK^2\oplus4\KK^0\oplus\KK^{-1}$\\ \hline
    \end{tabular}
    \caption{Irreducible decompositions of state tensor spaces (left table) and constitutive tensor spaces (right table).}
  \end{table}  
The challenge of future sections will be to explicitly construct the harmonic decompositions corresponding to these structures. It should be pointed out that as soon as the harmonic structure involves multiple spaces of the same order, the explicit decomposition is not uniquely defined \cite{GSS89}. As can be seen from the previous tables, this is the case for all the  constitutive tensor spaces considered in this study.

\section{Harmonic decomposition: methodology}\

In this section, we present the geometric objects and methods that are at the core of our approach to decompose tensors. With the exception of  Proposition \ref{prop:TrGam} which must be adapted, the results provided in this section are valid for 2D and 3D physical spaces. The 3D situation will be detailed in a future contribution, and we will focus here only on the 2D case.

\subsection{The harmonic decomposition}

Let $\VV$ and $\WW$ be two vector spaces, a map $\phi:\VV\rightarrow\WW$ is said to be $\ode$-equivariant if
\ben
\forall \vG\in\ode,\ \forall \vv\in\VV,\ \vG\star\phi(\vv)=\phi(\vG\star\vv).
\een
An explicit harmonic decomposition $\phi$ of a tensor $\vT\in\mathbb{T}^{n}$ is an $\ode$-equivariant linear  isomorphism between  a direct sum of harmonic spaces $\VV\simeq\bigoplus\KK^{k_{i}}$  and the space $\mathbb{T}^{n}$:
\begin{align*}
\phi: \VV\simeq\bigoplus\KK^{k_{i}}&\rightarrow \mathbb{T}^{n}\\
(\alpha,\ldots,\vK^{n})&\mapsto \vT=\phi(\alpha,\ldots,\vK^{n})
\end{align*}
Since this isomorphism is $\ode$-equivariant \cite{AKO16}, it satisfies the following property:
\begin{equation}\label{eq:EqHarDec}
\forall \vG\in\ode,\  \vG\star \vT=\vG\star \phi(\alpha,\ldots,\vK^{n})= \phi(\alpha,\ldots,\vG\star\vK^{n})
\end{equation}
which means that \emph{rotating} $\vT$ is equivalent to rotating the elements of its decomposition,  and vice versa.  
Since the transformations of the harmonic components (i.e. the elements of the harmonic decomposition) are elementary (cf. Equations \eqref{eq:TrsIrr}), it is generally easier to study the transformations of the harmonic components rather than the ones of the full tensor. In this view,  Equation \eqref{eq:EqHarDec} provides an explicit link between these two representations. The challenge is to obtain such an explicit expression for $\phi$.


\subsection{A three-step methodology}
Consider two spaces of state tensors denoted by $\mathbb{E}$ and $\mathbb{F}$. The (linear) constitutive law is an element  $\vT\in\mathcal{L}(\mathbb{E},\mathbb{F})$. In the present context, $\vT$ represents the constitutive tensor of which we want to obtain the harmonic decomposition.  The construction of a Clebsch-Gordan Harmonic Decomposition (in abbreviated form CGHD) of $\vT$ is obtained using the following procedure:
\begin{description}
\item[1) State Tensor Harmonic Decomposition (STHD)] Choose and compute a harmonic decomposition for  elements $\vv\in\mathbb{E}$ and $\ww\in\mathbb{F}$. This decomposition implies the definition of \emph{harmonic embedding} operators.  From these operators, we get a family of orthogonal projectors that will be used to decompose $\vT$;
\item[2) \red{Intermediate Block Decomposition} (\red{IBD})] Consider an element $\vT\in\mathcal{L}(\mathbb{E},\mathbb{F})$ which represents the constitutive tensor of which we want to obtain the harmonic decomposition. The choice of a STHD and the use of the associated projectors induce a decomposition of $\mathcal{L}(\mathbb{E},\mathbb{F})$ into "blocks". This decomposition, that  will be referred to as the \red{Intermediate Block Decomposition}, is not irreducible;

\item [3) Clebsch-Gordan Harmonic Decomposition (CGHD)] Each elementary block of the \red{Intermediate Block Decomposition} belongs to a space $\KK^p\otimes\KK^q$, the harmonic structure of which is known by the Clebsch-Gordan formula. The use of harmonic embeddings allows us to break down  each block into irreducible tensors.
\end{description}

The combination of the last two steps provides the Clebsch-Gordan Harmonic Decomposition of $\vT\in\mathcal{L}(\mathbb{E},\mathbb{F})$. The resulting decomposition  is a particular explicit harmonic decomposition of $\vT$  which is compatible with the harmonic decompositions of $\vv$ and $\ww$. This decomposition is uniquely defined\footnote{This claimed uniqueness is ensured by the property that in 2D the tensor product of irreducible harmonic spaces decomposes, as detailed in the tables of Lemma \ref{lem-Gordan}, into a direct sum of irreducible spaces of distinct orders. Even if the Clesbch-Gordan formula are different, this property is also valid in 3D \cite{JCB78}.}  by the choice of a particular form of the harmonic decompositions for the spaces $\mathbb{E}$ and $\mathbb{F}$.  It has to be noted that different choices for the decompositions of $\mathbb{E}$ and  $\mathbb{F}$ will lead to different  decompositions of $\vT$.

\subsection{Harmonic embeddings}

The isotypic decomposition of a tensor space $\mathbb{T}^{n}$ can be written in two different, but isomorphic, ways: 
\begin{equation}\label{eq:dec-iso}
  \mathbb{T}^{n} \simeq \bigoplus_{k=-1}^{n}\bigoplus_{l=0}^{p_{k}} \HH^{(n,k)}_{l}\simeq \bigoplus_{k=-1}^{n}\bigoplus_{l=0}^{p_{k}} \KK^{k}_{l}
\end{equation}
in which the involved spaces are:
\begin{description}
\item[$\TT^{n}$]a space of $n$th-order tensors with given index symmetries, it is the tensor space we want to decompose;
\item[$\HH^{(n,k)}$]a subspace of $\TT^{n}$ isomorphic to $\KK^{k}$, it is the embedding space for the elements belonging to $\KK^{k}$;
\item[$\KK^{k}$]a space of $k$th-order harmonic tensors, the elements of this space are used to parametrize the harmonic decomposition of $\vT\in\mathbb{T}^{n}$.
\end{description}
In the first decomposition of Equation \eqref{eq:dec-iso}, $\HH^{(n,k)}$ is a subspace of $\TT^{n}$ isomorphic to $\KK^{k}$ with $k\leq n$, while in the second one $\KK^{k}$ is a subspace of $\GG^{k}$.  In both decompositions, the first direct sum is on the order of the harmonic space, while the other one concerns the summation of the different spaces of the same order.  
The space $\HH^{(n,k)}$ serves as an intermediate space to embed a tensor $\vK\in\KK^{k}$ into $\vT\in\mathbb{T}^{n}$. Thus, the elements of $\HH^{(n,k)}$ are parametrized by elements from $\KK^{k}$, this parametrization is what we call the  \emph{harmonic embedding}. This technique will be used repeatedly in our work. As  will be seen in Section \ref{s:ClaEla},  this allows us to express the harmonic decompositions in terms of projection operators.\\

Let us consider more precisely the parametrization of $\HH^{(n,k)}$ by $\KK^{k}$. The connections of the different spaces are shown on the following diagram:
\begin{equation}
\xymatrix@!{
    \TT^{n} \ar[r]^{\mathrm{P}^{(n,k)}} & \HH^{(n,k)}\ar@<2pt>[d]^{\mathrm{\Pi}^{\{k,n\}}} \\
     \GG^{k} \ar[r]^{\mathrm{p}^{(k)}} & \KK^{k}\ar@<2pt>[u]^{\mathrm{\Phi}^{\{n,k\}}}
  }
\end{equation}
In this diagram, the associated mappings are:
\begin{description}
\item[$\mathrm{P}^{(n,k)}$]  a projector from $\TT^{n}$ to its subspace $\HH^{(n,k)}$, $\mathrm{P}^{(n,k)}$ is a $2n$th-order tensor;
\item[$\mathrm{p}^{(k)}$]        a projector from $\GG^{k}$ to its subspace $\KK^{k}$, $\mathrm{p}^{(k)}$  is a $2k$th-order tensor;
\item[$\mathrm{\Pi}^{\{k,n\}}$] a  projector from $\TT^{n}$ to its subspace $\KK^{k}$, $\mathrm{\Pi}^{\{k,n\}}$ is a $(n+k)$th-order tensor\footnote{\red{In the notation $\mathrm{\Pi}^{\{k,n\}}$ the order of the bracketed exponents is organized so that the left-most exponent indicates the tensor order of the image of the map, while the right-most exponent is the tensor order of the argument. This convention does not apply for in-parenthesis exponents such as those appearing in $\mathrm{P}^{(n,k)}$ which is a $2n$th-order tensor and not a $(n+k)$th-order one.} };
\item[$\mathrm{\Phi}^{\{n,k\}}$]a \emph{harmonic embedding} of  $\KK^{k}$ into $\TT^{n}$, $\mathrm{\Phi}^{\{n,k\}}$ is a $(n+k)$th-order tensor.
\end{description}
By definition,  we have the following  fundamental relations: 
\begin{equation}\label{eq:Carac1}
\mathrm{I}^{\HH^{(n,k)}}:=\mathrm{\Phi}^{\{n,k\}}\rdots{k}\mathrm{\Pi}^{\{k,n\}}\in\II^{2n},\quad
\mathrm{I}^{\KK^{k}}:=\mathrm{\Pi}^{\{k,n\}}\rdots{n}\mathrm{\Phi}^{\{n,k\}}\in\II^{2k}
\end{equation}
and, by construction,
\begin{equation}\label{eq:Carac2}
\mathrm{P}^{(n,k)}=\mathrm{I}^{\HH^{(n,k)}},\quad\mathrm{p}^{(k)}=\mathrm{I}^{\KK^{k}}.
\end{equation}

We have the remarkable property that all these different operators are known as soon as $\mathrm{\Phi}^{\{n,k\}}$  is determined. To see that, let us first define $\mathrm{\Phi}^{\{k,n\}}:=(\mathrm{\Phi}^{\{n,k\}})^{T}$ to be the transpose of $\mathrm{\Phi}^{\{n,k\}}$, i.e. the tensor which satisfies the following property:
\ben
\forall \vV\in\TT^{n},\ \forall\vv\in\KK^{k},\quad \langle\vV,\mathrm{\Phi}^{\{n,k\}}\rdots{k}\vv\rangle_{\TT^{n}}=\langle\mathrm{\Phi}^{\{k,n\}}\rdots{n}\vV,\vv\rangle_{\KK^{k}}.
\een
In terms of components, we have
\begin{equation}\label{def:Transp}
\left(\mathrm{\Phi}^{\{k,n\}}\right)_{i_{1}\ldots i_{n+k}}=\left(\mathrm{\Phi}^{\{n,k\}}\right)_{i_{n+1}\ldots i_{n+k}i_{1}\ldots i_{n}}.
\end{equation}
We have the following theorem:
\begin{thm}\label{thm:HarEmb}
 Let $\mathrm{\Phi}^{\{n,k\}}$ be a harmonic embedding of  $\KK^{k}$ into $\TT^{n}$. The operators $\mathrm{P}^{(n,k)}$, $\mathrm{p}^{(k)}$, $\mathrm{\Pi}^{\{k,n\}}$  defined above
can be expressed in terms of $\mathrm{\Phi}^{\{n,k\}}$ as follows:
\ben
\mathrm{\Pi}^{\{k,n\}}=\frac{1}{\gamma}\mathrm{\Phi}^{\{k,n\}},\quad
\mathrm{P}^{(n,k)}=\frac{1}{\gamma}\mathrm{\Phi}^{\{n,k\}}\rdots{k}\mathrm{\Phi}^{\{k,n\}},\quad
\mathrm{p}^{(k)}=\frac{1}{\gamma}\mathrm{\Phi}^{\{k,n\}}\rdots{n}\mathrm{\Phi}^{\{n,k\}}
\een
in which $\mathrm{\Phi}^{\{k,n\}}$ denotes the transpose of $\mathrm{\Phi}^{\{n,k\}}$, $\gamma$ is defined as
\ben
\gamma=\frac{\|\mathrm{\Phi}^{\{n,k\}}\rdots{k}\vv\|^2}{\|\vv\|^2},\  \vv\in\KK^{k}\setminus\{0\},
\een
and we adopt the convention that, in the case $k=0$, $ \rdots{0}=\otimes$.
\end{thm}
\begin{proof}
The proof is made by inserting the result of Lemma \ref{lem:InvPhi} into the relations provided by combining Equations \eqref{eq:Carac1} and \eqref{eq:Carac2}. Intermediate lemmas are provided in Appendix \ref{s:Pro31}. 
\end{proof}
The following proposition, which is demonstrated in Appendix \ref{s:Pro31}, gives another formula for the calculation of the constant $\gamma$:
\begin{prop}\label{prop:TrGam}
The constant $\gamma$ in Theorem \ref{thm:HarEmb}  can be calculated as
\ben
\gamma=\frac{1}{2} \tr \mathrm{M}
\een
in which $\mathrm{M}$ is the matrix of the linear map $\eta:\KK^k\rightarrow\KK^k$ defined by 
\ben
\eta(\vv)=\left(\mathrm{\Phi}^{\{k,n\}}\circ \mathrm{\Phi}^{\{n,k\}}\right)\udot\vv.
\een
\end{prop}

\section{Application to classical elasticity}\label{s:ClaEla}

The aim of the present section is to detail the method to the well-known situation of the fourth-order elasticity tensor. Results obtained here can be checked with the results available in the literature \cite{Via97,DA19}. The presentation is here mainly illustrative and will be extended in the next section to the more complicated situation of strain-gradient elasticity for which results are original and not available in the literature. Let  us begin by the construction of the harmonic decomposition of the state tensor space $\sdef$.

\subsection{Step 1: Decomposition of the state tensor space $\sdef$}\label{ss:decS2}

As indicated in Table \ref{tab:IrrDec}, $\sdef$ has the following harmonic structure:
\begin{equation} \label{eq:IsotypicOrdre2}
\sdef\simeq\HH^{(2,2)}\oplus\HH^{(2,0)}\simeq\KK^2\oplus\KK^0.
\end{equation}
The decomposition of an element $\dt\in\sdef$ is uniquely defined and it corresponds to the partition of $\dt$ into a deviatoric tensor $\dd \in \HH^{(2,2)}=\KK^2$  and a spherical one $\ds \in \HH^{(2,0)}$.  The identity tensor on $\sdef$  can be decomposed as the sum of the deviatoric and spherical projectors denoted $\qT{P}^{(2,2)}$ and $\qT{P}^{(2,0)}$:
\ben
\qT{I}^{\sdef}=\qT{P}^{(2,2)}+ \qT{P}^{(2,0)}.
\een

The structure of the associated harmonic embeddings are described on the following diagrams
\begin{equation}
\xymatrix@!{
    \sdef\ar[r]^{\qT{P}^{(2,2)}} & \HH^{(2,2)}\ar@<2pt>[d]^{\qT{I}^{\sdef}} \\
     \GG^{2} \ar[r]^{\qT{P}^{(2,2)}} & \KK^{2}\ar@<2pt>[u]^{\qT{I}^{\sdef}}
  }
,\qquad\quad 
  \xymatrix@!{
    \sdef\ar[r]^{\qT{P}^{(2,0)}} & \HH^{(2,0)}\ar@<2pt>[d]^{\dT{\Pi}^{\{0,2\}}} \\
     \KK^0 \ar[r]^{1} & \KK^{0}\ar@<2pt>[u]^{\dT{\Phi}^{\{2,0\}}}
  }.
\end{equation}
Let us build the spherical projector $\qT{P}^{(2,0)}$ first, the deviatoric projector $\qT{P}^{(2,2)}$ will then be deduced from it. Following the harmonic embedding method, the spherical part of $\dt$ will be parametrized by a scalar $\alpha\in\KK^0$. To construct the associated projector, let us first determine the embedding operator $\dT{\Phi}^{\{2,0\}}$:
\begin{equation}\label{eq:EmbK0}
\dT{\Phi}^{\{2,0\}}=\lambda  \id
\end{equation}
where $\lambda\in\RR\setminus\{0\}$ is a free scaling factor. In this situation, we have $\dT{\Phi}^{\{2,0\}} = \dT{\Phi}^{\{0,2\}}$. From Theorem \ref{thm:HarEmb}, $\dT{\Pi}^{\{0,2\}}$ has the following expression:
\ben
\dT{\Pi}^{\{0,2\}}=\frac{1}{\gamma}\dT{\Phi}^{\{0,2\}}=\frac{\lambda}{\gamma} \id=\frac{1}{2\lambda} \id,\quad\text{since}\quad \gamma=\|\dT{\Phi}^{\{2,0\}}\|^2=2\lambda^2.
\een
This results in the following expression for $\qT{P}^{(2,0)}$:
\begin{equation}\label{eq:P0}
\qT{P}^{(2,0)}=\dT{\Pi}^{\{2,0\}}\otimes\dT{\Phi}^{\{0,2\}}=2\dT{\Phi}^{\{2,0\}}\otimes\dT{\Phi}^{\{0,2\}}=\frac{1}{2}\id\otimes\id=\frac{1}{2}\qT{i}^{(4)}_1.
\end{equation}
We recognise here the well-known expression of the spherical projector.
Note that the choice of a specific $\lambda$ has no consequence on the expression of $ \qT{P}^{(2,0)}$.
The deviatoric projector can then be obtained:
\begin{equation}\label{eq:P2}
\qT{P}^{(2,2)}:=\qT{I}^{\sdef}-\qT{P}^{(2,0)}=\frac{1}{2}(\qT{i}^{(4)}_2+\qT{i}^{(4)}_3-\qT{i}^{(4)}_1).
\end{equation}
\begin{rem}
From Lemma \ref{lem:SymPro},  it appears that the tensors $\qT{P}^{(2,2)}$ and $\qT{P}^{(2,0)}$ can be considered as isotropic elasticity tensors in $\Elaq$. Interpreted as elements of  $(\Elaq,::)$, these tensors are associated to the following Gram matrix:
\begin{table}[H]
\center
\begin{tabular}{|c||c|c|}
\hline
$\qc$& $\qT{P}^{(2,2)}$ & $\qT{P}^{(2,0)}$   \\ \hline
$\qT{P}^{(2,2)}$ &$2$ & $0$    \\ \hline
$\qT{P}^{(2,0)}$ &$0$ &$1$  \\ \hline
\end{tabular}
\end{table}
\noindent Further, it can be noted that $\qT{P}^{(2,k)}\qc\qT{P}^{(2,k)}=\dim(\KK^{k}),\ k\in\{0,2\}.$
\end{rem}
Although the value of $\lambda$ has no importance for the expression of the projectors, it has some to construct an explicit parametrization of $\dt$ in terms of its harmonic components $\dd$ and $\alpha$. For our concern, the value of $\lambda$ will be chosen by imposing $\dT{\Pi}^{\{0,2\}}$ to be the standard trace operator, i.e.
\ben
\dT{\Pi}^{\{0,2\}}\dc\dt=\frac{1}{2\lambda}\id\dc\dt=\tr(\dt)
\een
which sets $\lambda$ to $\frac{1}{2}$.
Collecting all the previous observations we obtain the following results:
\begin{prop}\label{thm:DecT2}
There exists an $\ode$-equivariant isomorphism $\varphi$ between $\sdef$ and $\KK^{2}\oplus\KK^{0}$ such that
\ben
\dt=\dd+\ds=\dd+\dT{\Phi}^{\{2,0\}}\alpha=\varphi(\dd,\alpha)
\een
with $(\dd, \alpha) \in \KK^{2}\times\KK^{0}$ and  $\dT{\Phi}^{\{2,0\}}$ is such that           
\ben
\ds=\dT{\Phi}^{\{2,0\}}\alpha
\quad\text{with}\quad
\dT{\Phi}^{\{2,0\}}:=\frac{1}{2}\id.
\een 
Conversely, we have
\ben
\alpha=\dT{\Pi}^{\{0,2\}}:\dt=\tr(\dt)
\quad\text{with}\quad
\dT{\Pi}^{\{0,2\}}=\id.
\een
\end{prop}

\subsection{Step 2: \red{Intermediate Block Decomposition} of $\Elaq$}

In this subsection, and in the following ones, results will be provided under two forms:
\begin{enumerate}
\item the general one which involves embedding operators (as will be used in the next section);
\item the simplified one, since for the particular case of the elasticity tensor the operators have very simple expressions.
\end{enumerate} 

Using the family of projectors $(\qT{P}^{(2,2)},\qT{P}^{(2,0)})$ constructed from the harmonic decomposition of $\sdef$, we can demonstrate the following result:

\begin{prop}\label{prop:DecC}
The tensor $\qT{C}\in\Elaq$ admits the uniquely defined \red{Intermediate Block Decomposition} associated to the family of projectors $(\qT{P}^{(2,2)},\qT{P}^{(2,0)})$:
\begin{eqnarray}\label{eq:DecC}
\qT{C}&=&\qT{C}^{2,2}+2\left(\dT{h}^{2,0}\otimes\dT{\Phi}^{\{0,2\}}+\dT{\Phi}^{\{2,0\}}\otimes\dT{h}^{2,0}\right)+2\alpha^{0,0}\qT{P}^{(2,0)}\\
&=&\qT{C}^{2,2}+\left(\dT{h}^{2,0}\otimes\id+\id\otimes\dT{h}^{2,0}\right)+2\alpha^{0,0}\qT{P}^{(2,0)} \nonumber
\end{eqnarray}
in which $(\qT{C}^{2,2},\dT{h}^{2,0},\alpha^{0,0})$ are elements of $(\KK^2\otimes^s\KK^2)\times \KK^2\times\KK^0$.
Those elements are defined from $\qT{C}$ as follows:
\ben
\begin{cases}
\qT{C}^{2,2}=\qT{P}^{(2,2)}:\qT{C}:\qT{P}^{(2,2)},\\
\dT{h}^{2,0}=\qT{P}^{(2,2)}:\qT{C}:\dT{\Phi}^{\{2,0\}}=\frac{1}{2}\qT{P}^{(2,2)}:\qT{C}:\id,\\
 \alpha^{0,0}=\frac{1}{2}\qT{C}\qc\qT{P}^{(2,0)}=\frac{1}{4}\id:\qT{C}:\id,
\end{cases}
\een
where  $\dT{\Phi}^{\{2,0\}}$ is defined in Proposition \ref{thm:DecT2}. 
\end{prop}

\begin{proof}
Let consider $\dt$ and $\dt^\star$ in $\sdef$ such as $\dt^\star=\qT{C}:\dt$. The elements $\dt$ and $\dt^\star$ can be decomposed as
\ben
\dt=\dd+\ds,\ \dt^\star=\dd^\star+\ds^\star.
\een
Using the projectors  $(\qT{P}^{(2,2)},\qT{P}^{(2,0)})$, the following relations can be obtained:
\ben
\dt^\star=\qT{P}^{(2,2)}\dc\dt^\star+\qT{P}^{(2,0)}\dc\dt^\star\quad\text{and} \quad
\dt^\star=(\qT{C}:\qT{P}^{(2,2)})\dc\dt+(\qT{C}:\qT{P}^{(2,0)})\dc\dt.
\een
Through their combination the constitutive law  $\dt^\star=\qT{C}:\dt$ can be expressed as
\ben
\begin{cases}
\dd^{\star}=\qT{C}^{2,2}\dc\dd +\qT{C}^{2,0}\dc\ds\\
\ds^{\star}=\qT{C}^{0,2}\dc\dd +\qT{C}^{0,0}\dc\ds\\
\end{cases}
\quad\text{or using "matrix" notation}\quad
\begin{pmatrix}
\dd^{\star}\\
\ds^{\star}\\
\end{pmatrix}
=
\begin{pmatrix}
\qT{C}^{2,2} &\qT{C}^{2,0}  \\
\qT{C}^{0,2} &\qT{C}^{0,0} \\
\end{pmatrix}
\begin{pmatrix}
\dd\\
\ds\\
\end{pmatrix},
\een
in which:
\ben
\begin{cases}
\qT{C}^{2,2}=\qT{P}^{(2,2)}:\qT{C}:\qT{P}^{(2,2)},\\
\qT{C}^{2,0}=\qT{P}^{(2,2)}:\qT{C}:\qT{P}^{(2,0)},
\\
\qT{C}^{0,2}=\qT{P}^{(2,0)}:\qT{C}:\qT{P}^{(2,2)},
\\
\qT{C}^{0,0}=\qT{P}^{(2,0)}:\qT{C}:\qT{P}^{(2,0)}.
\end{cases}
\een
Expressed under this form, the symmetry of the constitutive law is not obvious. Using the relation
\ben
\qT{P}^{(2,0)}=2\dT{\Phi}^{\{2,0\}}\otimes\dT{\Phi}^{\{0,2\}}=\frac{1}{2}\id\otimes\id
\een
the previous relationships can be rewritten as follows
\ben
\begin{cases}
\qT{C}^{2,2}=\qT{P}^{(2,2)}:\qT{C}:\qT{P}^{(2,2)},\\
\qT{C}^{2,0}=2(\qT{P}^{(2,2)}:\qT{C}:\dT{\Phi}^{\{2,0\}})\otimes\dT{\Phi}^{\{0,2\}}=2\dT{h}^{2,0}\otimes\dT{\Phi}^{\{0,2\}},\\
\qT{C}^{0,2}=2\dT{\Phi}^{\{2,0\}}\otimes(\dT{\Phi}^{\{0,2\}}:\qT{C}:\qT{P}^{(2,2)})=2\dT{\Phi}^{\{2,0\}}\otimes\dT{h}^{2,0},\\
\qT{C}^{0,0}=2\alpha^{0,0}\qT{P}^{(2,0)},
\end{cases}
\een
where the following intermediate quantities have been introduced:
\ben
\dT{h}^{2,0}=\qT{P}^{(2,2)}:\qT{C}:\dT{\Phi}^{\{2,0\}},\quad \alpha^{0,0}=\dT{\Phi}^{\{0,2\}}:\qT{C}:\dT{\Phi}^{\{2,0\}}=\frac{1}{2}\qT{P}^{(2,0)}\qc\qT{C}.
\een
Therefore, we obtain:
\ban
\dt^{\star}&=&\dd^{\star}+\ds^{\star}\\
&=&\left(\qT{C}^{2,2}+2\dT{h}^{2,0}\otimes\dT{\Phi}^{\{0,2\}}\right):\dd+\left(2\dT{\Phi}^{\{2,0\}}\otimes\dT{h}^{2,0} +2\alpha^{0,0}\qT{P}^{(2,0)}\right):\ds\\
&=&\left(\qT{C}^{2,2}+2\dT{h}^{2,0}\otimes\dT{\Phi}^{\{0,2\}}\right):\left(\qT{P}^{(2,2)}:\dt\right)+\left(2\dT{\Phi}^{\{2,0\}}\otimes\dT{h}^{2,0} +2\alpha^{0,0}\qT{P}^{(2,0)}\right):\left(\qT{P}^{(2,0)}:\dt\right)\\
&=&\left(\qT{C}^{2,2}+2\left(\dT{\Phi}^{\{2,0\}}\otimes\dT{h}^{2,0}+\dT{h}^{2,0}\otimes\dT{\Phi}^{\{0,2\}}\right) +2\alpha^{0,0}\qT{P}^{(2,0)}\right):\dt
\ean
and hence, by identification:
\ben
\qT{C}=\qT{C}^{2,2}+2\left(\dT{\Phi}^{\{2,0\}}\otimes\dT{h}^{2,0}+\dT{h}^{2,0}\otimes\dT{\Phi}^{\{0,2\}}\right) +2\alpha^{0,0}\qT{P}^{(2,0)}.
\een
\end{proof}

\subsection{Step 3: Clebsch-Gordan Harmonic Decomposition of $\Elaq$}

The \red{Intermediate Block Decomposition} given by Equation \eqref{eq:DecC}  for the elasticity tensor is not irreducible, meaning that some of its components can further be decomposed.
To determine which components can be reduced and how to reduce them, the Clebsch-Gordan formula is essential.
By construction, $\qT{C}^{2,2}\in\KK^2\otimes^s\KK^2$, $\dT{h}^{2,0}\in\KK^2\otimes\KK^0$ and $\alpha^{0,0}\in\KK^0\otimes^s\KK^0$.  By applying the Clebsch-Gordan formula (cf. Lemma \ref{lem-Gordan}) to each one of these spaces we obtain:
\ben
\KK^2\otimes\KK^0\simeq\KK^2, \quad
\KK^0\otimes^s\KK^0\simeq\KK^0, \quad
\KK^2\otimes^s\KK^2\simeq \KK^4\oplus\KK^0,
\een
which indicates that:
\begin{itemize}
\item the components $(\dT{h}^{2,0},\alpha^{0,0})\in \KK^2\times\KK^0$ are already irreducible;
\item  the component $\qT{C}^{2,2}$ is reducible and can be decomposed into a scalar and a fourth-order harmonic tensor.
\end{itemize}
\noindent To proceed the decomposition of $\qT{C}^{2,2}$ consider the following lemma which is a direct application of Theorem \ref{thm:HarDom}:
\begin{lem}\label{lem:K2sK2}
Tensors $\qT{T}^{2,2}\in\KK^2\otimes^s\KK^2$ admit the uniquely defined harmonic decomposition
\ben
\qT{T}^{2,2}=\qT{H}+\frac{\alpha}{2}\qT{P}^{(2,2)},\quad\text{with}\  \qT{H}\in\KK^4,\ \alpha\in\RR
\een
conversely
\ben
\alpha=\qT{T}^{2,2}\qc\qT{P}^{(2,2)},\quad \qT{H}=\qT{T}^{2,2}-\frac{\alpha}{2}\qT{P}^{(2,2)}.
\een
\end{lem}

%
\begin{rem}
The structure of the projections is summed up in the following diagram:
\begin{equation}\label{eq:diagComm}
\xymatrix @!0 @R=3cm @C=4cm{
    \KK^2\otimes^s\KK^2\ar[r]^{\frac{1}{2}\qT{P}^{(2,2)}\otimes\qT{P}^{(2,2)}} & \HH^{(4,0)}\ar@<2pt>[d]^{\qT{\Pi}^{\{0,4\}}=\qT{P}^{(2,2)}} \\
     \KK^{0} \ar[r]^{1} & \KK^{0}\ar@<2pt>[u]^{\qT{\Phi}^{\{4,0\}}=\frac{1}{2}\qT{P}^{(2,2)}}
  }
\end{equation}
It can be observed that the method provides the intrinsic expression of the projector $\vP^{(4,0)}$ from  $\KK^2\otimes^s\KK^2$ onto $\HH^{(4,0)}$:
\ben
\vP^{(4,0)}=\frac{1}{2}\qT{P}^{(2,2)}\otimes\qT{P}^{(2,2)}.
\een
\end{rem}
The insertion of  the result of Lemma \ref{lem:K2sK2} in the \red{Intermediate Block Decomposition} demonstrates the following proposition:
\begin{prop}
The tensor $\qT{C}\in\Elaq$ admits the uniquely defined Clebsch-Gordan Harmonic Decomposition associated to the family of projectors $(\qT{P}^{(2,2)},\qT{P}^{(2,0)})$:
\begin{equation}\label{eq:DecHarEla}
\qT{C}=\qT{H}^{2,2}+\frac{\alpha^{2,2}}{2}\qT{P}^{(2,2)}+2\left(\dT{h}^{2,0}\otimes\dT{\Phi}^{\{0,2\}}+\dT{\Phi}^{\{2,0\}}\otimes\dT{h}^{2,0}\right)+2\alpha^{0,0}\qT{P}^{(2,0)}
\end{equation}
in which $(\qT{H}^{2,2},\dT{h}^{2,0}, \alpha^{2,2}, \alpha^{0,0})$ are elements of $\KK^4\times \KK^2\times\KK^0 \times \KK^0$ defined from $\qT{C}$ as follows:
\begin{center}
\begin{tabular}{|c|c|c|}
  \hline
 $\KK^{0}$ & $\KK^{2}$ & $\KK^{4}$  \\
  \hline
  $\alpha^{0,0}=\qT{P}^{(2,0)}\qc\qT{C}$  & &  \\ 
                                                              &$\dT{h}^{2,0}=\qT{P}^{(2,2)}:\qT{C}:\dT{\Phi}^{\{2,0\}}$ & \\
  $\alpha^{2,2}=\qT{C}^{2,2}::\qT{P}^{(2,2)}$&    &$\qT{H}^{2,2}=\qT{C}^{2,2}-\frac{\alpha^{2,2}}{2}\qT{P}^{(2,2)}$\\                                              
  \hline
\end{tabular}
\end{center}
where $\qT{C}^{2,2}=\qT{P}^{(2,2)}:\qT{C}:\qT{P}^{(2,2)}$. The projectors  $\qT{P}^{(2,2)}$ and $\qT{P}^{(2,0)}$ are defined in Equations \eqref{eq:P2} and \eqref{eq:P0} and $\dT{\Phi}^{\{2,0\}}$ is defined in Proposition \ref{thm:DecT2}.
\end{prop}

\subsection{Link with symmetry classes}

\red{The symmetry classes\footnote{\red{The \emph{symmetry class} of a tensor $\vT$ is the conjugacy class, in $\ode$, of its symmetry group $G_{\vT}:=\{\vG\in\ode,\ \vG\star\vT=\vT\}$.
The interested reader can find more details concerning this concept in the following references \cite{FV96,AKO16}.}} of bidimensional elasticity tensor have been determined in \cite{HZ96,BOR96,Via97}. The space $\Elaq$ is divided into the following $4$ symmetry classes:
\ben
\SymC(\Elaq) = \{[\ZZ_2], [\DD_2], [\DD_4], [\ode] \},
\een
in which $\SymC(\TT)$ is the set of symmetry classes of the space $\TT$.
These symmetry classes can be connected in the way indicated by the following diagram:
\ben
\xymatrix{
[\ZZ_2] \ar[r] & [\DD_2] \ar[r] & [\DD_4] \ar[r] & [\ode]
}
\een
As discussed in \cite{AKO16}, for a $\ZZ_{p}$-invariance, only the harmonic tensors of the order $kp, k\in\mathbb{N}$, are non-zero. The invariance with respect to $\ZZ^{\pi}_2$ cancels, as for it, the pseudo-scalars present in the decomposition. Consequently, the proposed harmonic decomposition becomes simpler as the number of symmetries increases :
\begin{itemize}
\item Symmetry class $[\ZZ_2], [\DD_2]$: the full expression; 
\item Symmetry class $[\DD_4]$:
\ben
\qT{C}=\qT{H}^{2,2}+\frac{\alpha^{2,2}}{2}\qT{P}^{(2,2)}+2\alpha^{0,0}\qT{P}^{(2,0)};
\een
\item Symmetry class $[\ode]$:
\ben
\qT{C}=\frac{\alpha^{2,2}}{2}\qT{P}^{(2,2)}+2\alpha^{0,0}\qT{P}^{(2,0)}.
\een
\end{itemize}
As discussed in \cite{Via97,AR16,DA19}, these simplifications are related to relations between invariants of the studied tensor. In these references, it is indicated, for instance, how to distinguish $[\ZZ_2]$ symmetry class from the $[\DD_2]$ one in terms of the harmonic tensors. For the sake of brevity, this aspect will not be  discussed further here. 
}

\section{Application to strain-gradient elasticity}

In this last section we apply the proposed methodology to the  fifth- and sixth-order elasticity tensors involved in strain-gradient elasticity.
To proceed these decompositions, in accordance with our method, the first step  consists in decomposing the \emph{state space} $\TT_{(ij)k}$ into a direct sum of $\ode$-irreducible spaces. 

\subsection{Decomposition of the state tensor space $\sgrd$}\label{ss:ProjT3}

The space $\sgrd$ has the following harmonic structure:
\ben\label{eq:HarDec2}
\sgrd\simeq \KK^{3}\oplus2\KK^{1}.
\een
Due to the multiplicity of  $\KK^{1}$ in the harmonic structure, the explicit harmonic decomposition is not uniquely defined \cite{GSS89}. The tensor in $\KK^{3}$  is canonically defined but there are multiple possibilities concerning the decomposition of the vector parts. Among the different possibilities, some have more physical content than others.  The one considered here consists in partitioning $\sgrd$ into a totally symmetric tensor ($\tT{S}\in\sym^3$) and a remainder, before proceeding to the harmonic decomposition of each part separately\footnote{Unlike the decomposition of  $\sdef$ which was direct, the decomposition of $\sgrd$ is two-step: 1) a splitting according to index symmetries of the tensor space; 2) the harmonic decomposition of the symmetric elementary part. This approach can be formalised (\emph{Schur-Weyl Harmonic Decomposition}), however this is not the subject of this contribution.}. \red{Another decomposition is proposed in Appendix \ref{sec:TypeII}.}
The process of the decomposition is described in the following diagram:
\ben
\xymatrix{
       &  \tT{T}\in\sgrd \ar[ld]_{\mathrm{Sym}}\ar[rd]^{\mathrm{Id}-\mathrm{Sym}} &  \\
  \tT{S}\in\sym^3\ar[d]^{\mathcal{H}}  &    & \tT{R}\in\HH^{r(3,1)} \ar[d]^{\mathcal{H}} \\
  (\tT{H},\V{v}^s)\in(\KK^{3}\times\KK^{1})&&\V{v}^r\in\KK^{1}}
\een
where $\mathrm{Sym}$ and $\mathcal{H}$ stand for the symmetrization and the harmonic decomposition processes, respectively. The space $\HH^{r(3,1)}$ appearing in this diagram is defined as
\ben
\HH^{r(3,1)}:=\{\tT{T}\in\sgrd |\ T_{ijk}+T_{jki}+T_{ikj}=0 \},\quad \dim\left(\HH^{r(3,1)}\right)=2.
\een
 In the strain-gradient literature $\tT{S}$ describes  the \emph{stretch-gradient part} of the strain-gradient tensor, while  $\tT{R}$ is the \emph{rotation-gradient} \cite{Min64,ME68}.

We have the following result:
\begin{thm}[\textbf{Harmonic decomposition of} $\TT_{(ij)k}$]\label{thm:DecT3}
There exists an $\ode$-equivariant isomorphism between $\TT_{(ij)k}$ and $\KK^{3} \oplus \KK^{1} \oplus \KK^{1}$ such that for  $\tT{H}\in \KK^{3}$ and $ (\V{v}^s, \V{v}^r) \in\KK^{1} \times\KK^{1}$,
\begin{equation}\label{eq:DecT3}
\tT{T} = \tT{H} + \qT{\Phi}^{s \{3,1\}} \udot \V{v}^s + \qT{\Phi}^{r \{3,1\}} \udot \V{v}^r,
\end{equation}
with $\Big(\qT{\Phi}^{s\{3,1\}},\qT{\Phi}^{r\{3,1\}} \Big)$ the harmonic embeddings of the form:
\begin{equation*}
\qT{\Phi}^{s\{3,1\}} = \frac{1}{4} \left(\qT{i}^{(4)}_1 + \qT{i}^{(4)}_2 + \qT{i}^{(4)}_3 \right) \quad;\quad
\qT{\Phi}^{r\{3,1\}} = \frac{1}{3} \left(2\qT{i}^{(4)}_1 - \qT{i}^{(4)}_2 - \qT{i}^{(4)}_3 \right),
\end{equation*}
in which $(\qT{i}^{(4)}_1, \qT{i}^{(4)}_2,\qT{i}^{(4)}_3)$ are the fourth-order elementary  isotropic tensors defined by Equation \eqref{eq:TIso4}.
Conversely, for any $\tT{T} \in \TT_{(ij)k}$,  $(\tT{H}, \V{v}^s, \V{v}^r) \in \KK^{3}\times\KK^{1} \times\KK^{1}$ are  defined from $\tT{T}$ as follows:
\begin{center}
\begin{tabular}{|c|c|}
  \hline
 $\KK^{1}$ & $\KK^{3}$  \\ \hline
 $ \V{v}^r=\qT{\Pi}^{r\{1,3\}}\tc\tT{T}$ &\\
 $ \V{v}^s=\qT{\Pi}^{s\{1,3\}}\tc\tT{T}$ &  $\tT{H}= \tT{T}- \qT{\Phi}^{s \{3,1\}} \udot \V{v}^s - \qT{\Phi}^{r \{3,1\}} \udot \V{v}^r$ \\
  \hline
\end{tabular}
\end{center}
with $\Big(\qT{\Pi}^{s\{1,3\}},\qT{\Pi}^{r\{1,3\}} \Big)$ the harmonic projectors of the form:
\begin{equation*}
\qT{\Pi}^{s\{1,3\}}=\frac{1}{3}\left(\qT{i}^{(4)}_3+\qT{i}^{(4)}_2+\qT{i}^{(4)}_1\right) \quad;\quad
\qT{\Pi}^{r\{1,3\}}=\frac{1}{2}\left(2\qT{i}^{(4)}_3-\qT{i}^{(4)}_2-\qT{i}^{(4)}_1\right).
\end{equation*}
\end{thm}

%
\begin{rem}
In the decomposition given by Equation \eqref{eq:DecT3},  $\V{v}^s$ represents the vector part of the \emph{stretch-gradient tensor} and $\V{v}^r$ represents the vector part of the \emph{rotation-gradient tensor}.
\end{rem}

\begin{rem}\label{rem:TypeIa}
\red{The decomposition is only valid for the \emph{Type II} formulation of strain-gradient elasticity. To obtain the same expressions for the \emph{Type I} formulation, that is for the decomposition of $\TT_{i(jk)}$,  the definition of $\qT{\Phi}^{r\{3,1\}}$  must be replaced by the following $\qT{\Phi}^{\sharp\ r\{3,1\}}$ 
\begin{equation}\label{eq:RelTI}
\qT{\Phi}^{\sharp, r\{3,1\}} = \frac{1}{3} \left(2\qT{i}^{(4)}_3 - \qT{i}^{(4)}_1 - \qT{i}^{(4)}_2 \right).
\end{equation}
This substitution applied to all future formulas will make it possible to obtain decompositions of the operators of the \emph{Type I} formulation from those of the \emph{Type II} formulation.}
\end{rem}

\begin{proof}
The decomposition of $\tT{T} \in \TT_{(ij)k}$ is two-step: 1) $\tT{T}$ is decomposed according to its index symmetries, then 2) each elementary part is decomposed into harmonic tensors.\\
\noindent \textbf{Step 1: Index symmetry splitting}\\
Any tensor $\tT{T} \in \TT_{(ij)k}$ can be decomposed into a complete symmetric tensor $\tT{S} \in\sym^3$ and a remainder $\tT{V}^{r}$ as:
\begin{equation} \label{dec-third-order}
\tT{T} = \tT{S} + \tT{V}^{r} 
\end{equation}
with
\begin{equation} \label{def-two-parts}
S_{ijk} = \frac{1}{3} (T_{ijk} + T_{ikj} + T_{jki}) \quad \text{and} \quad V_{ijk}^{r} = \frac{1}{3} (2 T_{ijk} - T_{ikj} - T_{jki}).
\end{equation}
It is direct to check that the tensor $\tT{V}^{r}$ belongs to $\mathbb{H}^{r(3,1)}$ and that the spaces $\sym^3$ and $\mathbb{H}^{r(3,1)}$ are in direct sum. \\

\noindent \textbf{Step 2: Harmonic decompositions}\\
Let us decompose $\tT{S}\in\sym^3$ and  $\tT{V}^{r}\in \mathbb{H}^{r(3,1)}$ into harmonic tensors. We consider this last case first. The following map
\ben
v^r_{i}\mapsto V^r_{ijk}:=\frac{1}{3} (2 \delta_{ij} v_k^r - \delta_{ik} v_j^r - \delta_{jk} v_i^r)
\een
is an embedding of $\KK^1$ into $\mathbb{H}^{r(3,1)}$, which can be rewritten in an intrinsic form as 
\begin{equation} \label{def-Phir31}
\tT{V}^{r} = \qT{\Phi}^{r\{3,1\}} \udot \V{v}^{r} \quad \text{with} \quad \qT{\Phi}^{r\{3,1\}} = \frac{1}{3} \Big( 2 \qT{i}^{(4)}_1 - \qT{i}^{(4)}_2 - \qT{i}^{(4)}_3 \Big) \in \mathbb{T}_{(ij)kl}.
\end{equation}
Since $\dim (\mathbb{H}^{r(3,1)}) = 2 = \dim (\KK^1)$, there is no other harmonic embedding to consider. The expression of the  projection from $\mathbb{H}^{r(3,1)}$ to $\KK^1$ is obtained using  Theorem \ref{thm:HarEmb}:
\ben
\qT{\Pi}^{r\{1,3\}}=\frac{1}{2}\left(2\qT{i}^{(4)}_3-\qT{i}^{(4)}_2-\qT{i}^{(4)}_1\right).
\een
Now consider the decomposition of $\tT{S}\in\sym^3$ into the sum of harmonic tensors. Using Lemma \ref{lem-Gordan} the harmonic structure of $\sym^3$ is known to be isomorphic to $\KK^3\oplus\KK^1$. Let $\mathbb{H}^{s(3,1)}$ be the space  of totally symmetric third-order tensors which are orthogonal to tensors in $\KK^{3}$:
\begin{equation}
\mathbb{H}^{s(3,1)}: = \left \{\tT{T} \in \sym^3 \mid \forall \tT{H}\in\KK^{3}, \tT{H}\tc\tT{T} = 0 \right \}.
\end{equation}
Since $\dim(\sym^3) = 4$ and $\dim(\KK^3) = 2$,  we deduce that $\dim (\mathbb{H}^{s(3,1)})= 2$. 
The following map:
\ben
\qT{\Phi}^{s\{3,1\}} \colon \KK^1 \to \mathbb{H}^{s(3,1)},\quad v^s_{i}\mapsto V^s_{ijk}:=\frac{1}{4} ( \delta_{ij} v_k^r + \delta_{ik} v_j^r + \delta_{jk} v_i^r)
\een
which can be rewritten in an intrinsic form as 
\begin{equation} \label{def-Phis31}
\tT{V}^{s} = \qT{\Phi}^{s\{3,1\}} \udot \V{v}^{s} \quad \text{with} \quad \qT{\Phi}^{s\{3,1\}} = \frac{1}{4} \left( \qT{i}^{(4)}_1 + \qT{i}^{(4)}_2 + \qT{i}^{(4)}_3 \right) \in \sym^{4}
\end{equation}
is an embedding of $\KK^1$ into $\mathbb{H}^{s(3,1)}$.
A direct application of Theorem \ref{thm:HarEmb} provides the expression of  the  projection from $\mathbb{H}^{s(3,1)}$ onto  $\KK^1$:
\ben
\qT{\Pi}^{s\{1,3\}}=\frac{1}{3}\left(\qT{i}^{(4)}_3+\qT{i}^{(4)}_2+\qT{i}^{(4)}_1\right).
\een
The tensor $\tT{H} \in \KK^3$ is then directly deduced: 
\ben
\tT{H}=\tT{S}-\qT{\Phi}^{s\{3,1\}} \udot \V{v}^{s}.
\een
\end{proof}

From the embedding operators involved in Theorem \ref{thm:HarEmb} a family of projectors can be deduced. 
\begin{prop}\label{cor:Proj3}
 Let $\sT{I}^{\sgrd}=\frac{1}{2}\left(\sT{i}^{(6)}_8+\sT{i}^{(6)}_{12}\right)$ be the identity tensor on $\sgrd$. 
The following tensors
\ben
\sT{P}^{(3,1r)}:=\frac{3}{2}(\qT{\Phi}^{r\{3,1\}}\udot\qT{\Phi}^{r\{1,3\}}),\quad
\sT{P}^{(3,1s)}:=\frac{4}{3}(\qT{\Phi}^{s\{3,1\}}\udot\qT{\Phi}^{s\{1,3\}}),\quad
\sT{P}^{(3,3)}:=\sT{I}^{\sgrd}-\sT{P}^{(3,1s)}-\sT{P}^{(3,1r)},
\een
where $\qT{\Phi}^{s\{1,3\}}$ and $\qT{\Phi}^{r\{1,3\}}$ are the transposes of $\qT{\Phi}^{s\{3,1\}}$ and $\qT{\Phi}^{r\{3,1\}}$, constitute a family  of orthogonal projectors on $\HH^{r(3,1)}$, $\HH^{s(3,1)}$ and $\KK^{3}$, respectively.
\end{prop}
\begin{proof}
The multiplication table of the family $(\sT{P}^{(3,3)},\sT{P}^{(3,1s)},\sT{P}^{(3,1r)})$ given below
\begin{table}[H]
\center
\begin{tabular}{|c||c|c|c|}
\hline
$\tc$& $\sT{P}^{(3,3)}$ & $\sT{P}^{(3,1s)}$ & $\sT{P}^{(3,1r)}$  \\ \hline
$\sT{P}^{(3,3)}$ &$\sT{P}^{(3,3)}$ & $\sT{0}$ &$\sT{0}$   \\ \hline
$\sT{P}^{(3,1s)}$ &$\sT{0}$ &$\sT{P}^{(3,1s)}$ & $\sT{0}$   \\ \hline
$\sT{P}^{(3,1r)}$ &$\sT{0}$&$\sT{0}$&$\sT{P}^{(3,1r)}$ \\   \hline
\end{tabular}
\end{table}
shows that  $(\sT{P}^{(3,3)},\sT{P}^{(3,1s)},\sT{P}^{(3,1r)})$ constitutes a family of orthogonal projectors with respect to the triple contraction.
\end{proof}

\begin{rem}\label{rem:TypeIb}
\red{To obtain the projectors for the \emph{Type I} formulation, that is for the decomposition of $\TT_{i(jk)}$,  the definition of $\sT{I}^{\sgrd}$  must be replaced by the following one
\ben
\sT{I}^{\TT_{i(jk)}}=\frac{1}{2}\left(\sT{i}^{(6)}_8+\sT{i}^{(6)}_{9}\right)
\een
By combining this identity with $\qT{\Phi}^{\sharp\ r\{3,1\}}$, the projectors $(\sT{P}^{\sharp(3,3)},\sT{P}^{(3,1s)},\sT{P}^{\sharp(3,1r)})$  adapted to the \emph{Type I} formulation are constructed.
}
\end{rem}

\begin{rem}
In Proposition \ref{cor:Proj3}, the transpose $\qT{\Phi}^{s\{1,3\}}$ of $\qT{\Phi}^{s\{3,1\}}$ (in the sense of Equation \eqref{def:Transp}) has been used.  However, since $\qT{\Phi}^{s\{3,1\}}\in \sym^4$, we have in this peculiar case $\qT{\Phi}^{s\{1,3\}}=\qT{\Phi}^{s\{3,1\}}$.  Even if it is superfluous in terms of algebra, the distinction has nevertheless been made here, and will be made in the following to ensure consistency of notation in our different expressions.
\end{rem}

\begin{rem}
From Lemma \ref{lem:SymPro}, the tensors $(\sT{P}^{(3,3)},\sT{P}^{(3,1s)},\sT{P}^{(3,1r)})$ can be considered as isotropic elasticity tensors in $\Elas$. Interpreted as elements of  $(\Elas,\rdots{6})$, these tensors are associated to the following Gram matrix:
\begin{table}[H]
\center
\begin{tabular}{|c||c|c|c|}
\hline
$\rdots{6}$& $\sT{P}^{(3,3)}$ & $\sT{P}^{(3,1s)}$ & $\sT{P}^{(3,1r)}$  \\ \hline
$\sT{P}^{(3,3)}$ &$2$ & $0$ &$0$   \\ \hline
$\sT{P}^{(3,1s)}$ &$0$ &$2$ & $0$ \\ \hline
$\sT{P}^{(3,1r)}$ &$0$&$0$&$2$ \\   \hline
\end{tabular}
\end{table}
\noindent on which it can be checked, with a slight abuse of notation, that $\sT{P}^{(3,k)}\rdots{6}\sT{P}^{(3,k)}=\dim(\KK^{k}),\ k\in\{1,3\}$.
\end{rem}

\subsection{Decomposition of $\Elas$}

We will  begin by considering this case, since the construction of  Clebsch-Gordan Harmonic Decomposition of $\sT{A}\in\Elas$ follows almost directly the method previously introduced for $\Elaq$.

\subsubsection{Clebsch-Gordan  Decomposition}
In the following proposition, the tensors $\qT{a}^{1s,3}$ and $\qT{a}^{1r,3}$ will denote the transposes of the tensors $\qT{a}^{3,1s}$ and $\qT{a}^{3,1r}$, respectively, in the sense defined by  Equation \eqref{def:Transp}.
\begin{prop}\label{prop:CGDec6}
The tensor $\sT{A}\in\Elas$ admits the uniquely defined \red{Intermediate Block Decomposition} associated to the family of projectors $(\sT{P}^{(3,3)},\sT{P}^{(3,1s)},\sT{P}^{(3,1r)})$:
\begin{eqnarray*}
\sT{A}&=&\sT{A}^{3,3}+\frac{16}{9}\qT{\Phi}^{s\{3,1\}}\udot\dT{a}^{1s,1s}\udot\qT{\Phi}^{s\{1,3\}}+\frac{9}{4}\qT{\Phi}^{r\{3,1\}}\udot\dT{a}^{1r,1r}\udot\qT{\Phi}^{r\{1,3\}}\\
&+&\frac{4}{3}\left(\qT{a}^{3,1s}\udot\qT{\Phi}^{s\{1,3\}} +\qT{\Phi}^{s\{3,1\}}\udot\qT{a}^{1s,3}\right)+\frac{3}{2}\left(\qT{a}^{3,1r}\udot\qT{\Phi}^{r\{1,3\}}+\qT{\Phi}^{r\{3,1\}}\udot\qT{a}^{1r,3}\right)\\
&+&2\left(\qT{\Phi}^{s\{3,1\}}\udot\dT{a}^{1s,1r}\udot\qT{\Phi}^{r\{1,3\}}+\qT{\Phi}^{r\{3,1\}}\udot\dT{a}^{1r,1s}\udot\qT{\Phi}^{s\{1,3\}}\right)
\end{eqnarray*}
in which $\sT{A}^{3,3}\in\KK^3\otimes^s\KK^3$, $(\qT{a}^{3,1s}, \qT{a}^{3,1r})\in(\KK^3\otimes\KK^1)^2$, $(\dT{a}^{1s,1s},\dT{a}^{1r,1r})\in(\KK^1\otimes^s\KK^1)^2$ and $\dT{a}^{1s,1r}\in\KK^1\otimes\KK^1$ and $\qT{\Phi}^{s\{3,1\}}, \qT{\Phi}^{r\{3,1\}}$ are defined in Proposition \ref{thm:DecT3}.
Those elements are defined from $\sT{A}$ as follows:
\begin{center}
\begin{tabular}{|c|c|c|}
  \hline
 $\TT^{2}$ & $\TT^{4}$ & $\TT^{6}$  \\
  \hline
  $\dT{a}^{1s,1s}:=\qT{\Phi}^{s\{1,3\}}\tc\sT{A}\tc\qT{\Phi}^{s\{3,1\}}$&$\qT{a}^{3,1s}:=\sT{P}^{(3,3)}\tc\sT{A}\tc\qT{\Phi}^{s\{3,1\}}$&$\sT{A}^{3,3}:=\sT{P}^{(3,3)}\tc\sT{A}\tc\sT{P}^{(3,3)}$\\
  $\dT{a}^{1s,1r}:=\qT{\Phi}^{s\{1,3\}}\tc\sT{A}\tc\qT{\Phi}^{r\{3,1\}}$  &$ \qT{a}^{3,1r}:=\sT{P}^{(3,3)}\tc\sT{A}\tc\qT{\Phi}^{r\{3,1\}}$&  \\ 
   $\dT{a}^{1r,1r}:=\qT{\Phi}^{r\{1,3\}}\tc\sT{A}\tc\qT{\Phi}^{r\{3,1\}}$ & &\\         \hline
\end{tabular}
\end{center}
\end{prop}

\begin{proof}
The proof follows the steps detailed in the proof of Proposition \ref{prop:DecC}, only the main points will be summed up here.
Starting from the decomposition introduced in Theorem \ref{thm:DecT3}, any $\tT{T}\in\sgrd$  decomposes as follows:
\ben
\tT{T}=\tT{H}+\tT{V}^{s}+\tT{V}^{r}.
\een
Using the projectors  $(\sT{P}^{(3,3)},\sT{P}^{(3,1s)},\sT{P}^{(3,1r)})$ and following the method used in the proof of Proposition \ref{prop:DecC}, the constitutive law can be brought to the following matrix form:
\ben
\begin{pmatrix}
\tT{H}^{\star}\\
\tT{V}^{\star s}\\
\tT{V}^{\star r}\\
\end{pmatrix}
=
\begin{pmatrix}
\sT{A}^{3,3} &\sT{A}^{3,1s} &\sT{A}^{3,1r} \\
\sT{A}^{1s,3} &\sT{A}^{1s,1s}&\sT{A}^{1s,1r}\\
\sT{A}^{1r,3} &\sT{A}^{1r,1s} &\sT{A}^{1r,1r} \\
\end{pmatrix}
\begin{pmatrix}
\tT{H}\\
\tT{V}^{s}\\
\tT{V}^{r}\\
\end{pmatrix},
\een
where the quantities are defined as:
\ben
\begin{cases}
\sT{A}^{3,3}:=\sT{P}^{(3,3)}\tc\sT{A}\tc\sT{P}^{(3,3)},\\
\sT{A}^{3,1s}:=\frac{4}{3}\qT{a}^{3,1s}\udot\qT{\Phi}^{s\{1,3\}},\\
\sT{A}^{3,1r}:=\frac{3}{2}\qT{a}^{3,1r}\udot\qT{\Phi}^{r\{1,3\}},\\
\sT{A}^{1s,3}:=\frac{4}{3}\qT{\Phi}^{s\{3,1\}}\udot\qT{a}^{1s,3},\\
\sT{A}^{1s,1s}:=\frac{16}{9}\qT{\Phi}^{s\{3,1\}}\udot\dT{a}^{1s,1s}\udot\qT{\Phi}^{s\{1,3\}},\\
\sT{A}^{1s,1r}:=2\qT{\Phi}^{s\{3,1\}}\udot\dT{a}^{1s,1r}\udot\qT{\Phi}^{r\{1,3\}},\\
\sT{A}^{1r,3}:=\frac{3}{2}\qT{\Phi}^{r\{3,1\}}\udot\qT{a}^{1r,3},\\
\sT{A}^{1r,1s}:=2\qT{\Phi}^{r\{3,1\}}\udot\dT{a}^{1r,1s}\udot\qT{\Phi}^{s\{1,3\}},\\
\sT{A}^{1r,1r}:=\frac{9}{4}\qT{\Phi}^{r\{3,1\}}\udot\dT{a}^{1r,1r}\udot\qT{\Phi}^{r\{1,3\}},\\
\end{cases}
\quad
\text{in which} 
\quad
\begin{cases}
\qT{a}^{3,1s}:=\sT{P}^{(3,3)}\tc\sT{A}\tc\qT{\Phi}^{s\{3,1\}},\\
\qT{a}^{3,1r}:=\sT{P}^{(3,3)}\tc\sT{A}\tc\qT{\Phi}^{r\{3,1\}},\\
\dT{a}^{1s,1s}:=\qT{\Phi}^{s\{1,3\}}\tc\sT{A}\tc\qT{\Phi}^{s\{3,1\}},\\
\dT{a}^{1s,1r}:=\qT{\Phi}^{s\{1,3\}}\tc\sT{A}\tc\qT{\Phi}^{r\{3,1\}},\\
\dT{a}^{1r,1r}:=\qT{\Phi}^{r\{1,3\}}\tc\sT{A}\tc\qT{\Phi}^{r\{3,1\}}.\\
\end{cases}
\een
Using matrix notation, we have
\ben
\begin{pmatrix}
\tT{H}^{\star}\\
\tT{V}^{\star s}\\
\tT{V}^{\star r}\\
\end{pmatrix}
=
\begin{pmatrix}
\sT{A}^{3,3}&\frac{4}{3}\qT{a}^{3,1s}\udot\qT{\Phi}^{s\{1,3\}} &\frac{3}{2}\qT{a}^{3,1r}\udot\qT{\Phi}^{r\{1,3\}} \\
\frac{4}{3}\qT{\Phi}^{s\{3,1\}}\udot\qT{a}^{1s,3}&\frac{16}{9}\qT{\Phi}^{s\{3,1\}}\udot\dT{a}^{1s,1s}\udot\qT{\Phi}^{s\{1,3\}} &2\qT{\Phi}^{s\{3,1\}}\udot\dT{a}^{1s,1r}\udot\qT{\Phi}^{r\{1,3\}} \\
\frac{3}{2}\qT{\Phi}^{r\{3,1\}}\udot\qT{a}^{1r,3} &2\qT{\Phi}^{r\{3,1\}}\udot\dT{a}^{1r,1s}\udot\qT{\Phi}^{s\{1,3\}}&\frac{9}{4}\qT{\Phi}^{r\{3,1\}}\udot\dT{a}^{1r,1r}\udot\qT{\Phi}^{r\{1,3\}} \\
\end{pmatrix}
\begin{pmatrix}
\tT{H}\\
\tT{V}^{s}\\
\tT{V}^{r}\\
\end{pmatrix}
\een.
\end{proof}

\subsubsection{Harmonic  Decomposition}

In the \red{Intermediate Block Decomposition} of $\sT{A}$ the non-harmonic tensors belong to $4$ different spaces:
\begin{itemize}
\item $\sT{A}^{3,3}\in\KK^3\otimes^s\KK^3\simeq\KK^6\oplus\KK^0$;
\item $(\qT{a}^{3,1s},\qT{a}^{3,1r})\in\KK^3\otimes\KK^1\simeq\KK^4\oplus\KK^2$;
\item $(\dT{a}^{1s,1s},\dT{a}^{1r,1r})\in\KK^1\otimes^s\KK^1\simeq\KK^2\oplus\KK^0$;
\item $\dT{a}^{1s,1r}\in\KK^1\otimes\KK^1\simeq\KK^2\oplus\KK^0\oplus\KK^{-1}$.
\end{itemize}
Their harmonic decompositions are provided in Appendix \ref{sec:HarmEmb}, the associated results are the following:\\
\noindent$\bullet$ Tensors $\sT{A}^{3,3}\in\KK^3\otimes^s\KK^3$ admit the uniquely defined harmonic decomposition
\ben
\sT{A}^{3,3}=\sT{H}+\frac{\alpha}{2}\sT{P}^{(3,3)},\ \text{where}\  \sT{H} \in \KK^6\  \text{and}\  \alpha \in \KK^0.  
\een
Conversely,
\ben
\alpha=\sT{A}^{3,3}\rdots{6}\sT{P}^{(3,3)} \quad\text{and}\  \quad \sT{H}=\sT{A}^{3,3}-\frac{\alpha}{2}\sT{P}^{(3,3)}.
\een
\noindent$\bullet$ Tensors $\qT{a}^{3,1}\in\KK^3\otimes\KK^1$ admit the uniquely defined harmonic decomposition
\ben
\qT{a}^{3,1}=\qT{H}+\sT{\Phi}^{\{4,2\}}:\dT{h},\quad\text{with}\  \qT{H}\in \KK^4,\   \dT{h} \in \KK^2\ \text{and}\ 
(\sT{\Phi}^{\{4,2\}}:\dT{h})_{ijkl}=\frac{1}{2}(h_{ij}\delta_{kl}+h_{ik}\delta_{jl}-h_{il}\delta_{jk}).
\een
Conversely,
\ben
\dT{h}=\tr_{14}\qT{a}^{3,1}\quad\text{and}\  \quad \qT{H}=\qT{a}^{3,1}-\sT{\Phi}^{\{4,2\}}:\dT{h}.
\een
The  notation $\tr_{ij}$ indicates that the contraction should be done on the $i$th and $j$th indices.\\
\noindent$\bullet$ Tensors $\dT{a}^{1,1}\in\KK^1\otimes\KK^1$ admit the uniquely defined harmonic decomposition
\ben
\dT{a}^{1,1}=\dd+\frac{\beta}{2}\dT{\epsilon}+\frac{\alpha}{2}\id,\  \text{with}\  \dT{d} \in \KK^2,\  (\alpha,\beta) \in \KK^0.
\een
Conversely,
\ben
\alpha=\dT{a}^{1,1}:\id, \quad \beta=\dT{a}^{1,1}:\dT{\epsilon}, \quad  \dd=\dT{a}^{1,1}-\frac{\beta}{2}\dT{\epsilon}-\frac{\alpha}{2}\id.
\een
In the following proposition, the notation $\vT^{T_{\alpha,\beta}}$ indicates a generalized transposition operation in which, in components, the first $\alpha$ indices are permuted with the last $\beta$ ones.

\begin{prop}[\textbf{Clebsch-Gordan Harmonic Decomposition of} $\sT{A}\in\Elas$]

The tensor $\sT{A}\in\Elas$ admits the uniquely defined  Clebsch-Gordan Harmonic Decomposition associated to the family of projectors $(\sT{P}^{(3,3)},\sT{P}^{(3,1s)},\sT{P}^{(3,1r)})$:

\ban
\sT{A}&=&\sT{H}^{3,3}+\frac{4}{3}\left(\qT{H}^{3,1s}\udot\qT{\Phi}^{s\{1,3\}} +\qT{\Phi}^{s\{3,1\}}\udot\qT{H}^{3,1s}\right)+\frac{3}{2}\left(\qT{H}^{3,1r}\udot\qT{\Phi}^{r\{1,3\}}+\qT{\Phi}^{r\{3,1\}}\udot\qT{H}^{3,1r}\right)\\
&+&\frac{16}{9}\qT{\Phi}^{s\{3,1\}}\udot\dT{h}^{1s,1s}\udot\qT{\Phi}^{s\{1,3\}}+\frac{9}{4}\qT{\Phi}^{r\{3,1\}}\udot\dT{h}^{1r,1r}\udot\qT{\Phi}^{r\{1,3\}}\\
&+&\frac{4}{3}\left(\left(\sT{\Phi}^{\{4,2\}}:\dT{h}^{3,1s} \right)\udot\qT{\Phi}^{s\{1,3\}}+\qT{\Phi}^{s\{3,1\}}\udot\left(\sT{\Phi}^{\{4,2\}}:\dT{h}^{3,1s}\right )^{T_{3,1}}\right)\\
&+&\frac{3}{2}\left(\left(\sT{\Phi}^{\{4,2\}}:\dT{h}^{3,1r}\right) \udot\qT{\Phi}^{r\{1,3\}}+\qT{\Phi}^{r\{3,1\}}\udot\left(\sT{\Phi}^{\{4,2\}}:\dT{h}^{3,1r} \right)^{T_{3,1}}\right)\\
&+&2\left(\qT{\Phi}^{s\{3,1\}}\udot\dT{h}^{1s,1r}\udot\qT{\Phi}^{r\{1,3\}}+\qT{\Phi}^{r\{3,1\}}\udot\dT{h}^{1r,1s}\udot\qT{\Phi}^{s\{1,3\}}\right)\\
&+&\frac{\alpha^{3,3}}{2}\sT{P}^{(3,3)}+\frac{2}{3}\alpha^{1s,1s}\sT{P}^{(3,1s)}+\frac{3}{4}\alpha^{1r,1r}\sT{P}^{(3,1r)}+\alpha^{1s,1r}\left(\qT{\Phi}^{s\{3,1\}}\udot\qT{\Phi}^{r\{1,3\}}+\qT{\Phi}^{r\{3,1\}}\udot\qT{\Phi}^{s\{1,3\}}\right)\\
&+&\beta^{1s,1r}\left(\qT{\Phi}^{s\{3,1\}}\udot\dT{\epsilon}\udot\qT{\Phi}^{r\{1,3\}}-\qT{\Phi}^{r\{3,1\}}\udot\dT{\epsilon}\udot\qT{\Phi}^{s\{1,3\}}\right),
\ean
in which $\sT{H}^{3,3}\in\KK^6$, $(\qT{H}^{3,1s}, \qT{H}^{3,1r})\in(\KK^4)^2$, $(\dT{h}^{3,1s},\dT{h}^{3,1r},\dT{h}^{1s,1r},\dT{h}^{1s,1s}\dT{h}^{1r,1r})\in(\KK^2)^5$, $(\alpha^{3,3},\alpha^{1s,1s},\alpha^{1r,1r},\alpha^{1r,1s})\in(\KK^0)^4$ and $\beta^{1r,1s}\in\KK^{-1}$.
Those elements are defined from $\sT{A}$ as follows:
\begin{center}
\begin{footnotesize}
\begin{tabular}{|c|c|c|c|c|}
  \hline
 $\KK^{-1}$ & $\KK^{0}$ & $\KK^{2}$ &  $\KK^{4}$ & $\KK^{6}$ \\
  \hline
  $\beta^{1s,1r} =\dT{a}^{1s,1r}:\Jd$  & $\alpha^{1s,1r} =\dT{a}^{1s,1r}:\id$ & $\dT{h}^{1s,1r}=\dT{a}^{1s,1r}:\qT{P}^{(2,2)}$ &&\\
                                                            & $\alpha^{1r,1r} =\dT{a}^{1r,1r}:\id$ & $\dT{h}^{1r,1r}=\dT{a}^{1r,1r}:\qT{P}^{(2,2)}$ &&\\
                                                            & $\alpha^{1s,1s} =\dT{a}^{1s,1s}:\id$ & $\dT{h}^{1s,1s}=\dT{a}^{1s,1s}:\qT{P}^{(2,2)}$ &&\\
                                                            &         &$\dT{h}^{3,1r}=\tr_{14}(\qT{a}^{3,1r})$&$\qT{H}^{3,1r}=\qT{a}^{3,1r}-\sT{\Phi}^{\{4,2\}}:\dT{h}^{3,1r}$&\\
           &                                                          &$\dT{h}^{3,1s}=\tr_{14}(\qT{a}^{3,1s})$&$\qT{H}^{3,1s}=\qT{a}^{3,1s}-\sT{\Phi}^{\{4,2\}}:\dT{h}^{3,1s}$&\\
      &$\alpha^{3,3}=\sT{A}^{3,3}\rdots{6}\sT{P}^{(3,3)}$&&&    $\sT{H}^{3,3}=\sT{A}^{3,3}-\frac{\alpha^{3,3}}{2}\sT{P}^{(3,3)}$\\ 
  \hline
\end{tabular}
\end{footnotesize}
\end{center}
in which the intermediate quantities are defined in Proposition \ref{prop:CGDec6}. 
\end{prop}

\begin{rem}
\red{Following Remarks \ref{rem:TypeIa} and \ref{rem:TypeIb}, the formula for strain-gradient elasticity expressed within type I formulation is directly obtained by substituting $\sT{P}^{\sharp(3,3)}$, $\sT{P}^{\sharp(3,1r)}$, and $\qT{\Phi}^{\sharp\ r\{3,1\}}$ for $\sT{P}^{(3,3)}$, $\sT{P}^{(3,1r)}$ and  $\qT{\Phi}^{r\{3,1\}}$  in the given formula.}
\end{rem}


\subsection{Link with symmetry classes}

\red{As demonstrated in \cite{ABB09,AKO16}, the space $\Elas$ is divided into  the 8 following symmetry classes:
\ben
 \SymC(\Elas)=\{[\ZZ_2],[\DD_2],[\ZZ_4],[\DD_4],[\ZZ_6],[\DD_6],[\sod],[\ode]\}
\een
These symmetry classes can be linked as shown in the following diagram:
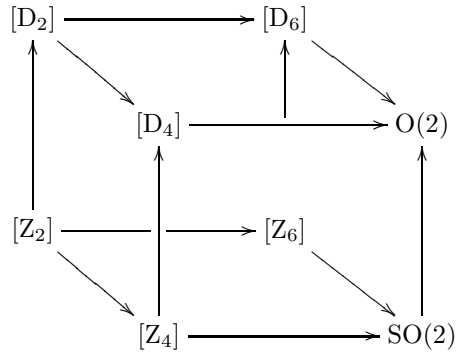
\ben
 \xymatrix {
    [\DD_2] \ar[rr]  \ar[dr] && [\DD_6] \ar[dr]  \\
    & [\DD_4] \ar[rr]  && \ode  \\
    [\ZZ_2] \ar[uu] \ar[rr] |!{[ur];[dr]}\hole \ar[dr] && [\ZZ_6] \ar[rd]\ar[uu] |!{[ul];[ur]}\hole \\
    & [\ZZ_4] \ar[uu]\ar[rr] && \sod\ar[uu]\\
  } 
\een
For the high-symmetry classes $[\ZZ_6]$,$[\DD_6]$,$[\sod]$,$[\ode]$ the harmonic decomposition reduces to 
\begin{itemize}
\item Symmetry class $[\ZZ_6]$:
\ban
\sT{A}&=&\sT{H}^{3,3}+\frac{\alpha^{3,3}}{2}\sT{P}^{(3,3)}+\frac{2}{3}\alpha^{1s,1s}\sT{P}^{(3,1s)}+\frac{3}{4}\alpha^{1r,1r}\sT{P}^{(3,1r)}+\alpha^{1s,1r}\left(\qT{\Phi}^{s\{3,1\}}\udot\qT{\Phi}^{r\{1,3\}}+\qT{\Phi}^{r\{3,1\}}\udot\qT{\Phi}^{s\{1,3\}}\right)\\
&+&\beta^{1s,1r}\left(\qT{\Phi}^{s\{3,1\}}\udot\dT{\epsilon}\udot\qT{\Phi}^{r\{1,3\}}-\qT{\Phi}^{r\{3,1\}}\udot\dT{\epsilon}\udot\qT{\Phi}^{s\{1,3\}}\right);
\ean
\item Symmetry class $[\DD_6]$:
\ban
\sT{A}&=&\sT{H}^{3,3}+\frac{\alpha^{3,3}}{2}\sT{P}^{(3,3)}+\frac{2}{3}\alpha^{1s,1s}\sT{P}^{(3,1s)}+\frac{3}{4}\alpha^{1r,1r}\sT{P}^{(3,1r)}+\alpha^{1s,1r}\left(\qT{\Phi}^{s\{3,1\}}\udot\qT{\Phi}^{r\{1,3\}}+\qT{\Phi}^{r\{3,1\}}\udot\qT{\Phi}^{s\{1,3\}}\right);
\ean
\item Symmetry class $[\sod]$:
\ban
\sT{A}&=&\frac{\alpha^{3,3}}{2}\sT{P}^{(3,3)}+\frac{2}{3}\alpha^{1s,1s}\sT{P}^{(3,1s)}+\frac{3}{4}\alpha^{1r,1r}\sT{P}^{(3,1r)}+\alpha^{1s,1r}\left(\qT{\Phi}^{s\{3,1\}}\udot\qT{\Phi}^{r\{1,3\}}+\qT{\Phi}^{r\{3,1\}}\udot\qT{\Phi}^{s\{1,3\}}\right)\\
&+&\beta^{1s,1r}\left(\qT{\Phi}^{s\{3,1\}}\udot\dT{\epsilon}\udot\qT{\Phi}^{r\{1,3\}}-\qT{\Phi}^{r\{3,1\}}\udot\dT{\epsilon}\udot\qT{\Phi}^{s\{1,3\}}\right);
\ean
\item Symmetry class $[\ode]$:
\ban
\sT{A}&=\frac{\alpha^{3,3}}{2}\sT{P}^{(3,3)}+\frac{2}{3}\alpha^{1s,1s}\sT{P}^{(3,1s)}+\frac{3}{4}\alpha^{1r,1r}\sT{P}^{(3,1r)}+\alpha^{1s,1r}\left(\qT{\Phi}^{s\{3,1\}}\udot\qT{\Phi}^{r\{1,3\}}+\qT{\Phi}^{r\{3,1\}}\udot\qT{\Phi}^{s\{1,3\}}\right).
\ean
\end{itemize}
In order to provide the harmonic decomposition for the other cases without ambiguity, additional tools, not discussed in this paper, are needed. We therefore defer this point to a future contribution.}

\subsection{Decomposition of $\Elac$}

In this last subsection, the Clebsch-Gordan Harmonic Decomposition of $\cT{M}\in\Elac$ is provided. Since $\Elac\simeq\mathcal{L}(\sgrd,\sdef)$, two different state tensor spaces are involved for constructing the decomposition.

\subsubsection{Clebsch-Gordan  Decomposition}

\begin{prop}\label{prop:CGDec5}
The tensor $\cT{M}\in\Elac$ admits the uniquely defined \red{Intermediate Block Decomposition} associated to the family of projectors $(\sT{P}^{(3,3)},\sT{P}^{(3,1s)},\sT{P}^{(3,1r)},\qT{P}^{(2,2)},\qT{P}^{(2,0)})$:
\ben
\cT{M}=\cT{M}^{2,3}+\frac{4}{3}\tT{m}^{2,1s}\udot\qT{\Phi}^{s\{1,3\}}+\frac{3}{2}\tT{m}^{2,1r}\udot\qT{\Phi}^{r\{1,3\}}+2\dT{\Phi}^{\{2,0\}}\otimes\tT{m}^{0,3}+\frac{8}{3}\left(\dT{\Phi}^{\{2,0\}}\otimes\V{\mu}^{0,1s}\right)\udot\qT{\Phi}^{s\{1,3\}}+3\left(\dT{\Phi}^{\{2,0\}}\otimes\V{\mu}^{0,1r}\right)\udot\qT{\Phi}^{r\{1,3\}} 
\een
in which $\cT{M}^{2,3}\in\KK^2\otimes\KK^3$, $(\tT{m}^{2,1s}, \tT{m}^{2,1r})\in (\KK^2\otimes\KK^1)^2$, $\tT{m}^{0,3}\in\KK^3$,  $(\V{\mu}^{0,1s},\V{\mu}^{0,1r})\in(\KK^1)^2$,  and $\qT{\Phi}^{s\{3,1\}}, \qT{\Phi}^{r\{3,1\}}$ and $\dT{\Phi}^{\{0,2\}}$ are defined, respectively, in Propositions \ref{thm:DecT3} and \ref{thm:DecT2}. 
Those elements are defined from $\cT{M}$ as follows:
\begin{center}
\begin{tabular}{|c|c|c|}
  \hline
 $\TT^{1}$ & $\TT^{3}$ & $\TT^{5}$  \\
  \hline
  $\V{\mu}^{0,1s}:=\dT{\Phi}^{\{0,2\}}:\cT{M}\tc\qT{\Phi}^{s\{3,1\}}$&$\tT{m}^{2,1s}:=\qT{P}^{(2,2)}:\cT{M}\tc\qT{\Phi}^{s\{3,1\}}$&$\cT{M}^{2,3}:=\qT{P}^{(2,2)}:\cT{M}\tc\sT{P}^{(3,3)}$\\
  $\V{\mu}^{0,1r}:=\dT{\Phi}^{\{0,2\}}:\cT{M}\tc\qT{\Phi}^{r\{3,1\}}$  &$ \tT{m}^{2,1r}:=\qT{P}^{(2,2)}:\cT{M}\tc\qT{\Phi}^{r\{3,1\}}$&  \\ 
    &$\tT{m}^{0,3}:=\dT{\Phi}^{\{0,2\}}:\cT{M}\tc\sT{P}^{(3,3)}$ & \\          \hline
\end{tabular}
\end{center}
\end{prop}
\begin{proof}
The proof follows the steps detailed in the proof of Proposition \ref{prop:DecC}, only the main points will be summed up here.
The constitutive law can be expressed in matrix notation as follows:
\ben
\begin{pmatrix}
\dd^{\star}\\
\ds^{\star}\\
\end{pmatrix}
=
\begin{pmatrix}
\cT{M}^{2,3} &\cT{M}^{2,1s} &\cT{M}^{2,1r}   \\
\cT{M}^{0,3} &\cT{M}^{0,1s} &\cT{M}^{0,1r}   \\
\end{pmatrix}
\begin{pmatrix}
\tT{H}\\
\tT{V}^{s}\\
\tT{V}^{r}\\
\end{pmatrix}
\een
with
\ben
\begin{cases}
\cT{M}^{2,3}:=\qT{P}^{(2,2)}:\cT{M}\tc\sT{P}^{(3,3)},\\
\cT{M}^{2,1s}:=\frac{4}{3}\tT{m}^{2,1s}\udot\qT{\Phi}^{s\{1,3\}},\\
\cT{M}^{2,1r}:=\frac{3}{2}\tT{m}^{2,1r}\udot\qT{\Phi}^{r\{1,3\}},\\
\cT{M}^{0,3}:=2\dT{\Phi}^{\{2,0\}}\otimes\tT{m}^{0,3},\\
\cT{M}^{0,1s}:=\frac{8}{3}\left(\dT{\Phi}^{\{2,0\}}\otimes\V{\mu}^{0,1s}\right)\udot\qT{\Phi}^{s\{1,3\}},\\
\cT{M}^{0,1r}:=3\left(\dT{\Phi}^{\{2,0\}}\otimes\V{\mu}^{0,1r}\right)\udot\qT{\Phi}^{r\{1,3\}},\\
\end{cases}
\quad
\text{in which}
\quad
\begin{cases}
\tT{m}^{2,1s}:=\qT{P}^{(2,2)}:\cT{M}\tc\qT{\Phi}^{s\{3,1\}},\\
\tT{m}^{2,1r}:=\qT{P}^{(2,2)}:\cT{M}\tc\qT{\Phi}^{r\{3,1\}},\\
\tT{m}^{0,3}:=\dT{\Phi}^{\{0,2\}}:\cT{M}\tc\sT{P}^{(3,3)},\\
\V{\mu}^{0,1s}:=\dT{\Phi}^{\{0,2\}}:\cT{M}\tc\qT{\Phi}^{s\{3,1\}},\\
\V{\mu}^{0,1r}:=\dT{\Phi}^{\{0,2\}}:\cT{M}\tc\qT{\Phi}^{r\{3,1\}}.\\
\end{cases}
\een
In matrix form, we have
\ben
\begin{pmatrix}
\dd^{\star}\\
\ds^{\star}\\
\end{pmatrix}
=
\begin{pmatrix}
\cT{M}^{2,3} &\frac{4}{3}\tT{m}^{2,1s}\udot\qT{\Phi}^{s\{1,3\}}&\frac{3}{2}\tT{m}^{2,1r}\udot\qT{\Phi}^{r\{1,3\}} \\
2\dT{\Phi}^{\{2,0\}}\otimes\tT{m}^{0,3}&\frac{8}{3}\left(\dT{\Phi}^{\{2,0\}}\otimes\V{\mu}^{0,1s}\right)\udot\qT{\Phi}^{s\{1,3\}} &3\left(\dT{\Phi}^{\{2,0\}}\otimes\V{\mu}^{0,1r}\right)\udot\qT{\Phi}^{r\{1,3\}}   \\
\end{pmatrix}
\begin{pmatrix}
\tT{H}\\
\tT{V}^{s}\\
\tT{V}^{r}\\
\end{pmatrix}.
\een
By construction,  it can directly be checked that
\ben
\cT{M}^{2,3}\in\KK^2\otimes\KK^3,\ (\tT{m}^{2,1s},\tT{m}^{2,1r})\in\KK^2\otimes\KK^1,\ \tT{m}^{0,3}\in\KK^0\otimes\KK^3,\ \text{and} (\V{\mu}^{0,1s},\V{\mu}^{0,1r})\in\KK^0\otimes\KK^1.
\een
\end{proof}

\subsubsection{Harmonic Decomposition}

In the \red{Intermediate Block Decomposition} of $\cT{M}$ the only non-harmonic tensors are $\cT{M}^{2,3}\in\KK^2\otimes\KK^3$ and  $(\tT{m}^{2,1s},\tT{m}^{2,1r})$ which belong to $\KK^2\otimes\KK^1$. Their harmonic decompositions are provided in Appendix \ref{sec:HarmEmb}, the associated results are the following:

\noindent$\bullet$ Tensors $\tT{m}^{2,1}\in\KK^2\otimes\KK^1$ admit the uniquely defined harmonic decomposition
\ben
\tT{m}^{2,1}=\tT{H}+\qT{\Phi}^{\{3,1\}}\udot\V{v},\quad  \text{where}\  \tT{H} \in \KK^3,\ \V{v} \in \KK^1,\ \text{and}\ 
(\qT{\Phi}^{\{3,1\}}\udot\V{v})_{ijk}=\frac{1}{2}(v_{i}\delta_{jk}+v_{j}\delta_{ik}-v_{k}\delta_{ij}).
\een
Conversely,
\ben
\V{v}=\tr_{13}\tT{m}^{2,1}\ \text{and}\   \tT{H}=\tT{m}^{2,1}-\qT{\Phi}^{\{3,1\}}\udot\V{v}.
\een
\noindent$\bullet$ Tensors $\cT{M}^{2,3}\in\KK^2\otimes\KK^3$ admit the uniquely defined harmonic decomposition
\ben
\cT{M}^{2,3}=\cT{H}+\sT{\Phi}^{\{5,1\}}\udot\V{v},\   \text{where}\  \cT{H} \in \KK^5,\  \V{v} \in \KK^1 
\een
and
\ban
4(\sT{\Phi}^{\{5,1\}}\udot\V{v})_{ijklm}&=&v_{i}(\delta_{jk}\delta_{lm}-\delta_{jl}\delta_{km}-\delta_{jm}\delta_{kl})-v_{j}\delta_{im}\delta_{kl}\\
&+&v_{k}(-2\delta_{ij}\delta_{lm}+\delta_{il}\delta_{jm}+2\delta_{im}\delta_{jl})+v_{l}\delta_{ik}\delta_{jm}+v_{m}\delta_{ij}\delta_{kl}
\ean
Conversely,
\ben
\V{v}=\tr_{12}\left(\tr_{13}\cT{M}^{2,3}\right)\ \text{and}\    \tT{H}=\cT{M}^{2,3}-\sT{\Phi}^{\{5,1\}}\udot\V{v}.
\een

Using these results we obtain:

\begin{prop}[\textbf{Clebsch-Gordan Harmonic Decomposition of} $\cT{M}\in\Elac$]
The tensor $\cT{M}\in\Elac$ admits the uniquely defined Clebsch-Gordan  Harmonic Decomposition associated to the family of projectors $(\sT{P}^{(3,3)},\sT{P}^{(3,1s)},\sT{P}^{(3,1s)},\qT{P}^{(2,2)},\qT{P}^{(2,0)})$:
\ban
\cT{M}&=&\cT{H}^{2,3}+\frac{4}{3}\tT{H}^{2,1s}\udot\qT{\Phi}^{s\{1,3\}}+\frac{3}{2}\tT{H}^{2,1r}\udot\qT{\Phi}^{r\{1,3\}}+2\dT{\Phi}^{\{2,0\}}\otimes\tT{H}^{0,3}+\sT{\Phi}^{\{5,1\}}\udot\V{v}^{2,3}\\
&+&\frac{4}{3}(\qT{\Phi}^{\{3,1\}}\udot\V{v}^{2,1s})\udot\qT{\Phi}^{s\{1,3\}}+\frac{3}{2}(\qT{\Phi}^{\{3,1\}}\udot\V{v}^{2,1r})\udot\qT{\Phi}^{r\{1,3\}}+
\frac{8}{3}\left(\dT{\Phi}^{\{2,0\}}\otimes\V{\mu}^{0,1s}\right)\udot\qT{\Phi}^{s\{1,3\}}\\
&+&3\left(\dT{\Phi}^{\{2,0\}}\otimes\V{\mu}^{0,1r}\right)\udot\qT{\Phi}^{r\{1,3\}}  
\ean
in which $\cT{H}^{2,3}\in\KK^5$, $(\tT{H}^{2,1s}, \tT{H}^{2,1r},\tT{H}^{0,3})\in(\KK^3)^3$, $(\V{v}^{2,3},\V{v}^{2,1s},\V{v}^{2,1r},\V{v}^{0,1s},\V{v}^{0,1r})\in(\KK^1)^5$.
Those elements are defined from $\cT{M}$ as follows:
\begin{center}
\begin{tabular}{|c|c|c|}
  \hline
 $\KK^{1}$ & $\KK^{3}$ & $\KK^{5}$  \\
  \hline
  $\V{v}^{0,1r} =\V{\mu}^{0,1r}$  & &  \\ 
  $\V{v}^{0,1s} =\V{\mu}^{0,1s}$ & & \\
  $\V{v}^{2,1r} =  \tT{m}^{2,1r}:\id$ &$\tT{H}^{2,1r}=\tT{m}^{2,1r}-\qT{\Phi}^{\{3,1\}}\udot\V{v}^{2,1r}$&\\
  $\V{v}^{2,1s}=\tT{m}^{2,1s}:\id$   &$\tT{H}^{2,1s}=\tT{m}^{2,1s}-\qT{\Phi}^{\{3,1\}}\udot\V{v}^{2,1s}$&\\
  &$\tT{H}^{0,3}= \tT{m}^{0,3}$&\\
  $\V{v}^{2,3}=\tr_{12}(\tr_{13}(\cT{M}^{2,3}))$&&$\cT{H}^{2,3}=\cT{M}^{2,3}-\sT{\Phi}^{\{5,1\}}\udot\V{v}^{2,3}$\\                                                                         
  \hline
\end{tabular}
\end{center}
in which the intermediate quantities are defined in Proposition \ref{prop:CGDec5}  and $\qT{\Phi}^{s\{3,1\}}, \qT{\Phi}^{r\{3,1\}}$ and $\dT{\Phi}^{\{0,2\}}$ are defined, respectively, in Propositions \ref{thm:DecT3} and \ref{thm:DecT2}. 
\end{prop}

\begin{rem}
\red{Following Remarks \ref{rem:TypeIa} and \ref{rem:TypeIb}, the formula for strain-gradient elasticity expressed within type I formulation is directly obtained by substituting $\sT{P}^{\sharp(3,3)}$, $\sT{P}^{\sharp(3,1r)}$, and $\qT{\Phi}^{\sharp\ r\{3,1\}}$   for $\sT{P}^{(3,3)}$, $\sT{P}^{(3,1r)}$ and  $\qT{\Phi}^{r\{3,1\}}$  in the given formula.}
\end{rem}

\subsection{Link with symmetry classes}
\red{As demonstrated in \cite{ADR15,AKO16}, the space $\Elac$ is divided into the 6 following symmetry classes:
 \begin{equation}
 \SymC(\Elac)=\{\triv,[\ZZ^{\pi}_2],[\ZZ_3],[\DD_3],[\DD_5],[\ode]\}
 \end{equation}
These symmetry classes can be linked as shown in the following diagram:
\begin{equation}
 \xymatrix {
    [\ZZ^{\pi}_2]\ar[rr]  \ar[dr] && [\DD_5] \ar[dr]  \\
    & [\DD_3] \ar[rr]  && [\ode]  \\
    [\triv] \ar[uu] \ar[dr] && \\
    & [\ZZ_3] \ar[uu]&& \\
  } 
\end{equation}
For the high-symmetries classes $[\DD_5]$,$[\ode]$ the harmonic decomposition  simplifies into
\begin{itemize}
\item  Symmetry class $[\DD_5]$:
\ben
\cT{M}=\cT{H}^{2,3};
\een
\item  Symmetry class $[\ode]$:
\ben
\cT{M}=\cT{0}.
\een
\end{itemize}
In order to provide the harmonic decomposition for the other cases without ambiguity, additional tools, not discussed in this paper, are needed. We therefore defer this point to a future contribution.}
\begin{rem}
As demonstrated in \cite{ADR15}, the set of symmetry classes for the complete constitutive law, that is of a triplet $\mathcal{E}:=\Big(\qT{C}, \cT{M}, \sT{A} \Big)\in \Sgrd$, is not the union of the symmetry classes of each tensor space considered separately. It can be shown, for example, that $[\ZZ_{5}]$ is a symmetry class of the complete  constitutive law, although it is not a symmetry classes of any individual tensor in the triplet  $\mathcal{E}$ .   We therefore refer interested readers to the reference \cite{ADR15} to see how to correctly combine the different results for each of the 14  symmetry classes of the 2D strain-gradient elasticity.
\end{rem}

\section{Conclusion}

In this paper the harmonic decomposition of the constitutive tensors appearing in the 2D Mindlin's strain-gradient elasticity has been investigated. Since no method available in the literature was considered satisfactory for the harmonic decomposition of higher-order tensors, a new harmonic decomposition, referred to here as the \textit{Clebsch-Gordan Harmonic Decomposition}, was proposed.
The main results of the paper are two-fold:
\begin{itemize}
\item the explicit 2D Clebsh-Gordan harmonic decompositions of:
\begin{itemize}
\item the fifth-order coupling tensor of strain-gradient elasticity;
\item the sixth-order elasticity tensor of strain-gradient elasticity.
\end{itemize}
\item the algorithm for the explicit Clebsh-Gordan harmonic decomposition for bidimensional tensors.
\end{itemize}
The Clebsch-Gordan algorithm is two-step and based on the explicit construction of the Clebsch-Gordan harmonic products. This approach, which shares some ideas with the one introduced by Zou in \cite{ZZD+01}, allows one to easily obtain orthogonal harmonic decompositions of higher-order tensors. Since the Clebsch-Gordan construction generates a new harmonic decomposition from a known one, the procedure can be iterated to obtain harmonic decompositions of arbitrary order tensors. The approach developed here in the 2D situation can be extended without any problem to the harmonic decomposition of 3D tensors. The study of this extension will be the object of a future contribution. 

\red{As we have already mentioned, the proposed method for decomposing tensors that we have introduced is very general and its applicability is by no means restricted to strain-gradient elasticity.  We believe that this method will find interesting applications beyond the one considered in the present contribution.}

\section*{Acknowledgements}
We thank Marc Olive for the multiple scientific discussions throughout the redaction of the manuscript.
The first and the second authors acknowledge the support of the French Agence Nationale de la Recherche (ANR), under grant ANR-17-CE08-0039 (project ArchiMatHOS).

\appendix
\renewcommand*{\thesection}{\Alph{section}}
 
\section{Proofs of Theorem 3.1 and Proposition 3.2}\label{s:Pro31}

This appendix is devoted to the formulation and the proofs of Lemmas required to demonstrate Theorem \ref{thm:HarEmb}  and  Proposition \ref{prop:TrGam}.
The main results of this appendix are Proposition \ref{prop:TrGam2} and Lemmas \ref{lem:SymPro} and \ref{prop:TrGam2}. The other Propositions are intermediate results necessary to demonstrate them.
%

\begin{prop}
If $n$ and $k$ are of the same parity, $\mathrm{\Phi}^{\{n,k\}}$ is an non-null isotropic tensor of order $k+n$, i.e.
\ben
\mathrm{\Phi}^{\{n,k\}}=\sum_{i}\lambda_{i}\vi^{(n+k)}_{i},
\een
otherwise $\mathrm{\Phi}^{\{n,k\}}$ is the null tensor.
\end{prop}
\begin{proof}
$\mathrm{\Phi}^{\{n,k\}}$ is an $\ode$-equivariant linear map between tensor spaces of order $n$ and $k$ with $n\geq k$, consequently:
\ben
\forall\vv\in\KK^k,\  \forall g\in\ode,\quad g\rayp{n}(\mathrm{\Phi}^{\{n,k\}}\rdots{k}\vv)=\mathrm{\Phi}^{\{n,k\}}\rdots{k}(g\rayp{k}\vv).
\een
Using the change of variables $\vv=g^{T}\rdots{k}\vv^{\star}$, we have
\ben
\forall\vv^\star\in\KK^k,\  \forall g\in\ode,\quad g\rayp{n}(\mathrm{\Phi}^{\{n,k\}}\rdots{k}(g^{T}\rayp{k}\vv^{\star}))=\mathrm{\Phi}^{\{n,k\}}\rdots{k}\vv^{\star}.
\een
Moreover,
\ban
g\rayp{n}(\mathrm{\Phi}^{\{n,k\}}\rdots{k}(g^{T}\rayp{k}\vv^{\star}))&=&g_{i_{1}j_{1}}\ldots g_{i_{n}j_{n}}\Phi_{j_{1}\ldots j_{n}k_{1}\ldots k_{k}}g^{T}_{k_{1}l_{1}}\ldots g^{T}_{k_{k}l_{k}}v^\star_{l_{1}\ldots l_{k}}\\
&=&g_{i_{1}j_{1}}\ldots g_{i_{n}j_{n}}g^{T}_{k_{1}l_{1}}\ldots g^{T}_{k_{k}l_{k}}\Phi_{j_{1}\ldots j_{n}k_{1}\ldots k_{k}}v^\star_{l_{1}\ldots l_{k}}\\
&=&g_{i_{1}j_{1}}\ldots g_{i_{n}j_{n}}g_{l_{1}k_{1}}\ldots g_{l_{k}k_{k}}\Phi_{j_{1}\ldots j_{n}k_{1}\ldots k_{k}}v^\star_{l_{1}\ldots l_{k}}\\
&=&(g\rayp{n+k}\mathrm{\Phi}^{\{n,k\}})\rdots{k}\vv^{\star}.
\ean
Then
\ben
\forall\vv^\star\in\KK^k,\  \forall g\in\ode,\quad (g\rayp{n+k}\mathrm{\Phi}^{\{n,k\}}- \mathrm{\Phi}^{\{n,k\}}) \rdots{k}\vv^{\star}=0
\een
which implies that
\ben
\forall g\in\ode,\quad g\rayp{n+k}\mathrm{\Phi}^{\{n,k\}} =\mathrm{\Phi}^{\{n,k\}}.
\een
Since the only isotropic tensor of odd order is the null tensor, $\mathrm{\Phi}^{\{n,k\}}$ is null if $n$ and $k$ are of different parity. If $n$ and $k$ are of the same parity, then $\mathrm{\Phi}^{\{n,k\}}$ is an isotropic tensor of order $n+k$, and thus can be expressed as a linear combination of elements of $\II^{(n+k)}$, i.e.
\ben
\mathrm{\Phi}^{\{n,k\}}=\sum_{i}\lambda_{i}\vi^{(n+k)}_{i}.
\een
\end{proof}

%

\begin{prop}
Consider $\mathrm{\Phi}^{\{n,k\}}\in\II^{(n+k)}$ and $\vv\in\KK^{k}$.  The image $\vV\in\HH^{(n,k)}$ of $\vv$ by $\mathrm{\Phi}^{\{n,k\}}$ has the following form:
\ben
\vV=\mathrm{\Phi}^{\{n,k\}}\rdots{k}\vv=\sum_{j}\lambda_{j}\varsigma_{j}\ast(\vi^{(n-k)}_{1}\otimes\vv), \quad \text{with}\  \varsigma_{j}\in\GS_{n}.
\een
\end{prop}

\begin{proof}
Since $\mathrm{\Phi}^{\{n,k\}}=\sum_{i}\lambda_{i}\vi^{(n+k)}_{i}$, $\vV$ has the following expression 
\ben
\vV=\mathrm{\Phi}^{\{n,k\}}\rdots{k}\vv=\sum_{i}\lambda_{i}\vi^{(n+k)}_{i}\rdots{k}\vv.
\een
Since $\vv\in\KK^{k}$, $\vv$ is totally symmetric and traceless, as such any term $\vi^{(n+k)}_{i}\rdots{k}\vv$  which implies contraction within $\vv$ disappears. The non-zero terms are, up to index permutation, those of the form $\vi^{(n-k)}_{1}\otimes\vv$, which gives the announced result.
\end{proof}

\begin{prop}\label{lem:PropNorm}
Let $\vV\in\HH^{(n,k)}$ be the image of $\vv\in\KK^{k}\setminus\{0\}$ by $\mathrm{\Phi}^{\{n,k\}}\in\II^{(n+k)}$.  There exists  $\gamma>0$ independent of $\vv$ such that
\ben
\|\vV\|^2=\gamma \|\vv\|^2.
\een
\end{prop}
\begin{proof}
As
\ben
\vV=\sum_{j}\lambda_{j}\varsigma_{j}\ast(\vi^{(n-k)}_{1}\otimes\vv) \quad \text{with}\  \varsigma_{j}\in\GS_{n},
\een
we remark that $\vV$ is null if and only if $\vv$ is  null.  So let us assume that $\vv\in\KK^{k}\setminus\{0\}$. Since $\|\vV\|^2$ and $\|\vv\|^2$ are strictly positive,  there exists $\gamma>0$ such $\|\vV\|^2=\gamma \|\vv\|^2$.  Let us show that $\gamma$ is independent of $\vv$. We consider the function $\rho:\KK^{k}\setminus\{0\}\to\RR^{+}$ defined by
\ben
\rho(\vv):=\frac{\|\mathrm{\Phi}^{\{n,k\}}\rdots{k}\vv\|^{2}}{\|\vv\|^2}
\een
in which the norms are the Frobenius norms associated with the dot product corresponding to the tensor order, i.e. $\rdots{k}$ for $k$th-order tensors.
Since $\KK^{k}$ is irreducible and its elements transform as vectors, any element of $\KK^{k}$ can be obtained from a non-null reference one, $\vv^{1}$, up to a scaling factor and up to a rotation. We can then write
\ben
\vv=\lambda g\rayp{k}\vv^{1},\quad \text{with}\ \lambda\in\RR^*,\ \vG\in\ode,
\een
obviously the scaling and the rotation transformations commute. As a consequence, $\gamma$ is independent of $\vv$ if the function $\rho$ is constant on $\KK^{k}\setminus\{0\}$.
Let us show that the function $\rho$ is constant on $\KK^k\setminus\{0\}$.  
\begin{itemize}
\item  We observe that
\ben
\forall(\lambda,\vv)\neq(0,\mathbf{0}),\quad \rho(\lambda\vv)=\frac{\|\mathrm{\Phi}^{\{n,k\}}\rdots{k}\lambda\vv\|^{2}}{ \|\lambda\vv\|^2}=\frac{\|\mathrm{\Phi}^{\{n,k\}}\rdots{k}\vv\|^{2}}{ \|\vv\|^2}=\rho(\vv).
\een
So $\rho$ is an homogeneous function of degree $0$, meaning that $\rho(\vv)$ is independent of the norm of $\vv$.
\item $\rho$ is an isotropic function, i.e. $\rho(\vG\star\vv)= \rho(\vv)$ for $\vv\neq 0$ and $\forall \vG\in\ode$. 
Firstly, $\mathrm{\Phi}^{\{n,k\}}$ is $\ode$-equivariant, hence 
\ben
\vv'=g\rayp{k}\vv \Rightarrow \vV'=g\rayp{n} \vV;
\een
secondly, norms are isotropic functions, as such
\ben
\forall \vG\in\ode,\ \rho(\vG\rayp{k}\vv)=\frac{\|\mathrm{\Phi}^{\{n,k\}}\rdots{k}(\vG\rayp{k}\vv)\|^{2}}{ \|\vG\rayp{k}\vv\|^2}=\frac{\|\vG\rayp{n}\vV\|^2}{ \|\vG\rayp{k}\vv\|^2}=\frac{\|\vV\|^2}{ \|\vv\|^2}=\rho(\vv).
\een
This result means that $\rho(\vv)$ is independent of the orientation of $\vv$.
\end{itemize}
Since the scaling and the rotation transformations commute, and since the function is constant for both actions considered separately, we have
\ben
\rho(\vv)=\rho(\lambda\vG\rayp{k}\vv^{1})=\rho(\vv^{1})=:\gamma, \quad\forall\vv\in\KK^{k}\setminus\{0\}.
\een
Hence the constant $\gamma$ is independent of the considered vector $\vv\in\KK^{k}$.
\end{proof}

The following proposition gives a method to compute the  parameter $\gamma$, it corresponds to  Proposition \ref{prop:TrGam} of  Section 3.

\begin{prop}\label{prop:TrGam2}
The constant $\gamma$  defined in Proposition \ref{lem:PropNorm} can be also calculated as
\ben
\gamma=\frac{1}{2} \tr \mathrm{M}
\een
in which $\mathrm{M}$ is the matrix of the linear map $\eta:\KK^k\rightarrow\KK^k$ defined by 
\ben
\eta(\vv)=\left(\mathrm{\Phi}^{\{k,n\}}\circ \mathrm{\Phi}^{\{n,k\}}\right)\udot\vv.
\een
\end{prop}

\begin{proof}
Let $\vv \in \KK^k\setminus\{0\}$ and $\vV: = \Phi^{\{n,k\}} \rdots{k}\vv$. From Proposition \ref{lem:PropNorm}, there exists $\gamma > 0$ such that $\|\vV\|^2=\gamma \|\vv\|^2$.
Let us consider
\ben
\mathrm{M}^{(k,k)}=\mathrm{\Phi}^{\{k,n\}}\rdots{n}\mathrm{\Phi}^{\{n,k\}}
\een
which can be  considered  as a symmetric second-order tensor on $\KK^{k}$.  By introducing $\mathrm{I}^{\KK^k}$ the second-order identity tensor on $\KK^{k}$, the relation $\|\vV\|^2=\gamma \|\vv\|^2$ can be expressed as:
\ben
 \vv\udot\mathrm{M}^{(k,k)}\udot\vv=\gamma\vv\udot\mathrm{I}^{\KK^k}\udot\vv.
\een
Differentiating with respect to $\vv$ we obtain
\ben
\left(\mathrm{M}^{(k,k)}-\gamma\mathrm{I}^{\KK^k}\right)\udot\vv=0.
\een
By considering a specific basis for $\KK^k$, the former relation can be reformulated in terms of matrix,
\ben
 \left([\mathrm{M}]-\gamma[\mathrm{I}]\right)\udot[\vv]=0.
\een
Since $\vv\neq 0$ the previous relation shows that $\gamma$ is an eigenvalue of $[\mathrm{M}]$. This result can be refined by noticing the following points:
\begin{itemize}
\item $\forall k>0,\ \dim\left(\KK^k\right)=2$, as such $[\mathrm{M}]$ has at most 2 different eigenvalues ;
\item Since by construction $\mathrm{M}^{(k,k)}$ is an isotropic tensor, $[\mathrm{M}]$ is proportional to $[\mathrm{I}]$.  A a result, $\gamma$ is  a double eigenvalue.
\end{itemize}
As a consequence,
\ben
\gamma=\frac{1}{2} \tr [\mathrm{M}].
\een
\end{proof}

The next result is fundamental to demonstrate Theorem \ref{thm:HarEmb}. 

\begin{lem}\label{lem:InvPhi}
The tensor  $\mathrm{\Phi}^{\{n,k\}}$ is invertible, and its inverse $\left(\mathrm{\Phi}^{\{n,k\}}\right)^{-1}$ has the following expression:
\ben
\left(\mathrm{\Phi}^{\{n,k\}}\right)^{-1}=\frac{1}{\gamma}\mathrm{\Phi}^{\{k,n\}}
\een
in which $\mathrm{\Phi}^{\{k,n\}}$ denotes the transpose of $\mathrm{\Phi}^{\{n,k\}}$ as defined by Equation \eqref{def:Transp}.
\end{lem}

\begin{proof}
Let $\vv\in\KK^{k}$ and $\vV= \mathrm{\Phi}^{\{n,k\}}\rdots{k}\vv$, we have
\ben
\|\vV\|^2=\langle \mathrm{\Phi}^{\{n,k\}}\rdots{k}\vv,\mathrm{\Phi}^{\{n,k\}}\rdots{k}\vv\rangle_{\HH^{(n,k)}}=\langle\vv,\mathrm{\Phi}^{\{k,n\}}\rdots{n}\mathrm{\Phi}^{\{n,k\}}\rdots{k}\vv\rangle_{\KK^{k}}.
\een
From Lemma \ref{lem:PropNorm} we have:
\ben
\|\vV\|^2=\gamma\|\vv\|^2,
\een
then
\ben
\langle\vv,\mathrm{\Phi}^{\{k,n\}}\rdots{n}\mathrm{\Phi}^{\{n,k\}}\rdots{k}\vv\rangle_{\KK^{k}}=\langle\vv,\gamma\vI^{\KK^{k}}\rdots{k}\vv\rangle_{\KK^{k}}
\een
which is equivalent to
\ben
\langle\vv,(\mathrm{\Phi}^{\{k,n\}}\rdots{n}\mathrm{\Phi}^{\{n,k\}}-\gamma\vI^{\KK^{k}})\udot\vv\rangle_{\KK^{k}}=0.
\een
Since this true for all $\vv\in\KK^{k}$, we deduce that
\ben
\mathrm{\Phi}^{\{k,n\}}\rdots{n}\mathrm{\Phi}^{\{n,k\}}=\gamma\vI^{\KK^{k}}
\een
and then
\ben
\left(\mathrm{\Phi}^{\{n,k\}}\right)^{-1}=\frac{1}{\gamma}\mathrm{\Phi}^{\{k,n\}}.
\een
\end{proof}

%

Finally,  the next result will be used  in Appendix B to directly characterize some harmonic embeddings.

\begin{lem}\label{lem:SymPro} 
Consider $\TT^{n}$ a space of $n$th-order tensors  and let $\mathbb{H}^{(n,k)}$ be a subspace of $\TT^n$ isomorphic to a harmonic space $\KK^k$ with $k \leq n$. Then the projector  $\mathrm{P}^{(n,k)}$ from $\TT^{n}$ onto $\HH^{(n,k)}$ belongs to 
$\mathcal{L}^s(\TT^{n},\TT^{n})$.
\end{lem}
\begin{proof}
By definition, $\mathrm{P}^{(n,k)}$ is a linear map from $\TT^{n}$ to $\TT^{n}$. We can therefore consider it as a tensor of order $2n$. It remains to verify the major symmetry, i.e. that  $\mathrm{P}^{(n,k)}=\left(\mathrm{P}^{(n,k)}\right)^T$.
Since
\ben
\mathrm{P}^{(n,k)}=\frac{1}{\gamma}\mathrm{\Phi}^{\{n,k\}}\rdots{k}\left(\mathrm{\Phi}^{\{n,k\}}\right)^{T},
\een
the major symmetry is verified.
\end{proof}


 
\section{Clebsch-Gordan harmonic embeddings}\label{sec:HarmEmb}

This appendix is devoted to the demonstrations of the fundamental explicit harmonic decompositions associated with the embeddings $\KK^{p+q}\oplus\KK^{|p-q|}\hookrightarrow\KK^p\otimes\KK^q$ used in the main part of the article.  Since the embedding of $\KK^{p+q}$ into $\KK^p \otimes \KK^q$   is trivial, the problem reduces to the determination of the unique embedding of $\KK^{|p-q|}$ into $\KK^p\otimes\KK^q$.  Knowing the algebraic characterisation of $\KK^p\otimes\KK^q$, this question can be reformulated in terms of linear algebra.

We have the following result:
\begin{thm}\label{thm:HarDom}
For $n\geq1$, let $\vP^{(n,n)}$ be the tensor associated to the projector from $\TT^{n}$ onto $\KK^{n}$ and consider $\vT^{n,n}$ an element of $\mathcal{L}^s(\KK^{n},\KK^{n})\simeq\KK^{n}\otimes^s\KK^{n}$. The tensor $\vT^{n,n}$ can be parametrized as follows:
\ben
\vT^{n,n}=\vK+\frac{\alpha}{2}\vP^{(n,n)},\quad (\vK,\alpha)\in\KK^{2n}\times\KK^{0},
\een
in such a way that $\vT^{n,n}\rdots{2n}\vP^{(n,n)}=\alpha$.
\end{thm}
\begin{proof}
First, by using the Clebsch-Gordan formula, it is known that $\mathcal{L}^s(\KK^{n},\KK^{n})\simeq\KK^{2n}\oplus\KK^{0}$. As a consequence  $\vT^{n,n}\in\mathcal{L}^s(\KK^{n},\KK^{n})$ can be written as 
\ben
\vT^{n,n}=\vK+\alpha \mathrm{\Phi}^{\{2n,0\}}
\een
with $\vK\in\KK^{2n},\alpha\in\KK^{0}$ and $\mathrm{\Phi}^{\{2n,0\}}$ is an isotropic tensor of order $2n$ element of $\mathcal{L}^s(\TT^{n},\TT^{n})$.
As a direct consequence of Lemma $\ref{lem:SymPro}$, it can be observed that $\vP^{(n,n)}$ is also an isotropic tensor of order $2n$ element of $\mathcal{L}^s(\TT^{n},\TT^{n})$.
Since $\dim\left(\HH^{(n,0)}\right)=1$,
\ben
\mathrm{\Phi}^{\{2n,0\}}=\lambda\vP^{(n,n)}.
\een
The scaling factor $\lambda$ is determined so that  $\vT^{n,n}\rdots{2n}\vP^{(n,n)}=\alpha$. We have
\ben
\vT^{n,n}\rdots{2n}\vP^{(n,n)}=\alpha\lambda\vP^{(n,n)}\rdots{2n}\vP^{(n,n)}.
\een
Since $\vP^{(n,n)}=\mathrm{I}^{\KK^{n}}$, $(\vP^{(n,n)}\rdots{2n}\vP^{(n,n)})=\dim(\KK^{n})=2$. Then $\vT^{n,n}\rdots{2n}\vP^{(n,n)}=2\alpha\lambda$ and we deduce that $\lambda=\frac{1}{2}$.
\end{proof}

\begin{cor}[\textbf{Decomposition of} $\KK^3\otimes^s\KK^3$]
Elements $\sT{A}^{3,3}$ of $\KK^3\otimes^s\KK^3$ can be decomposed as follows:
\ben
\sT{A}^{3,3}=\sT{H}+\frac{\alpha^{3,3}}{2}\sT{P}^{(3,3)}, \quad \text{where}\  \alpha^{3,3}=\sT{A}^{3,3}\rdots{6}\sT{P}^{(3,3)}
\een
with $\sT{P}^{(3,3)}$ defined in Proposition \ref{cor:Proj3}.
\end{cor}

%
%

\begin{lem}[\textbf{Decomposition of} $\KK^3\otimes\KK^1$ ]\label{lem:K3xK1}
There exists an $\ode$-equivariant isomorphism between $\KK^3\otimes\KK^1$ and $\KK^{4}\oplus\KK^{2}$ such that for any $\qT{T}^{3,1}\in\KK^3\otimes\KK^1$,
\ben
\qT{T}^{3,1}=\qT{H}+\sT{\Phi}^{\{4,2\}}:\dT{h},
\een
where  $(\qT{K}, \dT{h}) \in \KK^{4}\times\KK^{2}$ with $\dT{h}=\tr_{14}\qT{T}^{3,1}$ and $\sT{\Phi}^{\{4,2\}}$ is such that      
\ben
(\sT{\Phi}^{\{4,2\}}:\dT{h})_{ijkl}=\frac{1}{2}(h_{ij}\delta_{kl}+h_{ik}\delta_{jl}-h_{il}\delta_{jk})
\een
and
\ben
(\sT{\Phi}^{\{4,2\}})_{ijklmn}=\frac{1}{2}\left(\delta_{kl}P^{(2,2)}_{ijmn}+\delta_{jl}P^{(2,2)}_{ikmn}-\delta_{jk}P^{(2,2)}_{ilmn}\right).
\een   
Above, $\qT{P}^{(2,2)}$ is the standard deviatoric projector defined in Equation \eqref{eq:P2}.
Moreover, the inverse of $\sT{\Phi}^{\{4,2\}}$ is given by
\ben
(\sT{\Pi}^{\{2,4\}})_{ijklmn}=\sT{\Phi}^{\{2,4\}}_{klmnij}=\frac{1}{2}\left(\delta_{mn}P^{(2,2)}_{klij}+\delta_{ln}P^{(2,2)}_{kmij}-\delta_{lm}P^{(2,2)}_{knij}\right).
\een
\end{lem}

\begin{proof}
From the Clebsch-Gordan formula in 2D, it is known that
\ben
\qT{T}^{3,1}=\qT{H}+\sT{\Phi}^{\{4,2\}}:\dT{h},\qquad\text{with}\ \qT{H}\in\KK^4\ \text{and}\ \dT{h}\in\KK^2.
\een
It can be checked that $\KK^4\subset\KK^3\otimes\KK^1$ and  the  question is  the embedding of $\KK^2$ into $\KK^3\otimes\KK^1$.
Up to a scaling factor, there is a unique way to do so. The embedding can be determined by solving a linear system. A general embedding of $\KK^2$ into $\overset{4}{\otimes}\RR^2$ is given by:
\ben
(\sT{\Phi}^{\{4,2\}}:\dT{h})_{ijkl}=a_{1}h_{ij}\delta_{kl}+a_{2}h_{ik}\delta_{jl}+a_{3}h_{il}\delta_{jk}+a_{4}h_{jk}\delta_{il}+a_{5}h_{jl}\delta_{ik}+a_{6}h_{kl}\delta_{ij}.
\een
It can be checked that, in $\RR^2$, for a generic $\dT{h}\in\KK^2$, the family of tensors $\{h_{ij}\delta_{kl},h_{ik}\delta_{jl},h_{il}\delta_{jk},h_{jk}\delta_{il},h_{jl}\delta_{ik},h_{kl}\delta_{ij}\}$ is not free. For instance, 
\ben
\begin{cases}
h_{kl}\delta_{ij}=h_{il}\delta_{jk}+h_{jk}\delta_{il}-h_{ij}\delta_{kl}\\
h_{jl}\delta_{ik}=h_{il}\delta_{jk}+h_{jk}\delta_{il}-h_{ik}\delta_{jl}.
\end{cases}
\een
However, the same family restricted to its four first elements is  free. As a consequence, we will consider the following free parametrization:
\ben
(\sT{\Phi}^{\{4,2\}}:\dT{h})_{ijkl}=b_{1}h_{ij}\delta_{kl}+b_{2}h_{ik}\delta_{jl}+b_{3}h_{il}\delta_{jk}+b_{4}h_{jk}\delta_{il}.
\een
For $\sT{\Phi}^{\{4,2\}}:\dT{h}$ to belong to $\KK^3\otimes\KK^1$ the following conditions have to be satisfied:
\begin{enumerate}
\item complete symmetry with respect to $(ijk)$:
\ben 
\varsigma_{(123)}\star (\sT{\Phi}^{\{4,2\}}:\dT{h})-(\sT{\Phi}^{\{4,2\}}:\dT{h})=\qT{0},\quad \text{with}\ \varsigma_{(123)}\in\GS_{3};
\een
\item traceless with respect to $(ijk)$:
\ben
 \id:(\sT{\Phi}^{\{4,2\}}:\dT{h})=\dT{0}.
\een
\end{enumerate}
As a consequence,  $\sT{\Phi}^{\{4,2\}}:\dT{h}\in\KK^3\otimes\KK^1$ has the following form
 \ben
(\sT{\Phi}^{\{4,2\}}:\dT{h})_{ijkl}=b_{1}(h_{ij}\delta_{kl}+h_{ik}\delta_{jl}-h_{il}\delta_{jk}).
\een
The value of $b_{1}$ is determined by the closure condition:
\ben
\tr_{14}\left(\sT{\Phi}^{\{4,2\}}:\dT{h}\right)=\dT{h}
\een
which implies that $b_{1}=\frac{1}{2}$. So, at the end:  
\begin{equation}\label{eq:Phi42}
(\sT{\Phi}^{\{4,2\}}:\dT{h})_{ijkl}=\frac{1}{2}(h_{ij}\delta_{kl}+h_{ik}\delta_{jl}-h_{il}\delta_{jk}).
\end{equation}
To obtain $\sT{\Phi}^{\{4,2\}}$, observe that 
\ben
\dT{h}=\qT{I}^{\KK^{2}}\dc\dT{h}=\qT{P}^{(2,2)}\dc\dT{h}.
\een
Inserting this relation into Equation \eqref{eq:Phi42}, we obtain the expression of $\Phi^{(4,2)}$.
Since $\|\sT{\Phi}^{\{4,2\}}:\dT{h}\|^2=\|\qT{h}\|^2$, the application of  Theorem \ref{thm:HarEmb} gives
\ben
\sT{\Pi}^{\{2,4\}}=\left(\sT{\Phi}^{\{4,2\}}\right)^{T},
\een
and a direct computation allows us to check that
\ben
\sT{\Pi}^{\{2,4\}}\qc\left(\sT{\Phi}^{\{4,2\}}:\dT{h}\right)=\tr_{14}\left(\sT{\Phi}^{\{4,2\}}:\dT{h}\right)=\dT{h}.
\een
The structure of the harmonic embedding is summed-up in the following diagram
\begin{equation}\label{diag:Phi42}
\xymatrix@!{
    \KK^3\otimes\KK^1\ar[r]^{} & \HH^{(4,2)}\ar@<2pt>[d]^{\qT{\Pi}^{(2,4)}} \\
     \GG^{2} \ar[r]^{\qT{P}^{(2,2)}} & \KK^{2}\ar@<2pt>[u]^{\qT{\Phi}^{\{4,2\}}}
  }
\end{equation}

\end{proof}
For the three next lemmas, the proofs follow the same lines and will not be detailed.
\begin{lem}[\textbf{Decomposition of }$\KK^2\otimes\KK^3$ ]\label{lem:K2xK3}
There exists an $\ode$-equivariant isomorphism between $\KK^2\otimes\KK^3$ and $\KK^{5}\oplus\KK^{1}$ such that for any $\cT{T}^{2,3}\in\KK^2\otimes\KK^3$,
\ben
\cT{T}^{2,3}=\cT{H}+\sT{\Phi}^{\{5,1\}}\udot\V{v}
\een
where $(\cT{H}, \V{v}) \in \KK^{5}\times\KK^{1}$, with $\V{v}=\tr_{12}\left(\tr_{13}\cT{T}^{2,3}\right)$ and $\sT{\Phi}^{\{5,1\}}$ is  such that
\ban
(\sT{\Phi}^{\{5,1\}}\udot\V{v})_{ijklm}&=&\frac{1}{4}\left(v_{i}(\delta_{jk}\delta_{lm}-\delta_{jl}\delta_{km}-\delta_{jm}\delta_{kl})-v_{j}\delta_{im}\delta_{kl}\right.\\
&+&\left. v_{k}(-2\delta_{ij}\delta_{lm}+\delta_{il}\delta_{jm}+2\delta_{im}\delta_{jl})+v_{l}\delta_{ik}\delta_{jm}+v_{m}\delta_{ij}\delta_{kl}\right),
\ean
Intrinsically ,
\ben
\sT{\Phi}^{\{5,1\}}=\frac{1}{4}\left(\sT{i}^{(6)}_{1}-2\sT{i}^{(6)}_{3}+\sT{i}^{(6)}_{5}+\sT{i}^{(6)}_{8}-\sT{i}^{(6)}_{11}+2\sT{i}^{(6)}_{12}+\sT{i}^{(6)}_{13}-\sT{i}^{(6)}_{14}-\sT{i}^{(6)}_{15}\right).
\een   
The inverse of $\sT{\Phi}^{\{5,1\}}$ is given by
\ben
(\sT{\Pi}^{\{1,5\}})_{ijklmn}=(\sT{\Phi}^{\{5,1\}})_{jklmni}.
\een
\end{lem}

\begin{lem}[\textbf{Decomposition of} $\KK^2\otimes\KK^1$ ]\label{lem:K2xK1}
There exists an $\ode$-equivariant isomorphism between $\KK^2\otimes\KK^1$ and $\KK^{3}\oplus\KK^{1}$ such that
\ben
\tT{T}^{2,1}=\tT{H}+\qT{\Phi}^{\{3,1\}}\udot\V{v}, 
\een
with $(\tT{K}, \V{v}) \in \KK^{3}\times\KK^{1}$ with $\V{v}=\tr_{13}\tT{T}^{2,1}$. Here, the tensor $\qT{\Phi}^{\{3,1\}}$ coincides with  the standard deviatoric projector $\qT{P}^{(2,2)}$ defined in Equation \eqref{eq:P2}.
The inverse of $\qT{\Phi}^{\{3,1\}}$ is given by
\ben
(\qT{\Pi}^{\{1,3\}})_{ijkl}=(\qT{P}^{(2,2)})_{lijk}.
\een
\end{lem}


\begin{lem}\label{lem:K-1}
There exists an  $\ode$-equivariant embedding of $\KK^{-1}$ into $\GG^2$ such that
\ben
\dT{T}^{1,1}=\beta\dT{\Phi}^{\{2,-1\}}
\een
where $\beta\in\RR$, $\dT{\Phi}^{\{2,-1\}}=\frac{1}{2}\dT{\epsilon}$, with  $\beta=\dT{T}^{1,1}:\dT{\epsilon}$. Consequently, $\dT{\Pi}^{\{-1,2\}}=\dT{\epsilon}$ and the projector $\qT{P}^{(2,-1)}$ from $\GG^2$ onto $\HH^{(2,-1)}$ has the following expression:
\ben
\qT{P}^{(2,-1)}=\frac{1}{2}\dT{\epsilon}\otimes\dT{\epsilon}.
\een
\end{lem}

\section{Alternative decomposition of  $\sgrd$}\label{sec:TypeII}

\red{Due to the multiplicity of  $\KK^{1}$ in the harmonic structure of $\sgrd$, its explicit harmonic decomposition is not uniquely defined \cite{GSS89}. As a consequence of the \textit{Clebsch-Gordan Harmonic Algorithm}, considering another explicit decomposition for the state tensor space will affect the explicit harmonic decomposition of the associated constitutive tensors. In this appendix we propose another decomposition for the tensors in $\sgrd$ which is used in strain-gradient elasticity \cite{ELS+20} and in strain-gradient plasticity \cite{SF96,HF97}. 
In the modified strain-gradient theory developed in \cite{ELS+20}, the strain-gradient tensor is reduced to the gradient of the hydrostatic strain tensor.  In contrast, in the  incompressible strain-gradient plasticity theory \cite{SF96,HF97}, this contribution is constrained to be zero.
In both theories,  the constraint relates to a \emph{derivative} of  the strain tensor. Consequently, we introduce here a decomposition of $\sgrd$  based on the derivation of the harmonic decomposition of symmetric second-order tensors. This approach, which is related to Mindlin \emph{Type II }formulation, can be described as follows :
\ben
\xymatrix{
       &  \dT{T}\in\sdef \ar[ld]^{\mathcal{H}}\ar[rd]^{\mathcal{H}} &  \\
  \dT{T}^{d}\in\KK^{2}\ar[d]^{\otimes\V{\nabla}}  &    & \dT{T}^{h}\in\KK^{0}\ar[d]^{\otimes\V{\nabla}} \\
  (\tT{H},\V{v}^d)\in(\KK^{3}\times\KK^{1})&&\V{v}^h\in\KK^{1}}
\een
Associated to this decomposition we have the following result:
\begin{thm}[\textbf{Alternative Harmonic Decomposition of} $\TT_{(ij)k}$]\label{thm:DecT3II}
There exists an $\ode$-equivariant isomorphism between $\TT_{(ij)k}$ and $\KK^{3} \oplus \KK^{1} \oplus \KK^{1}$ such that for  $\tT{H}\in \KK^{3}$ and $ (\V{v}^d, \V{v}^h) \in\KK^{1} \times\KK^{1}$,
\begin{equation}\label{eq:DecT3II}
\tT{T} = \tT{H} + \qT{\Phi}^{d \{3,1\}} \udot \V{v}^d + \qT{\Phi}^{h \{3,1\}} \udot \V{v}^h,
\end{equation}
with $\Big(\qT{\Phi}^{d\{3,1\}},\qT{\Phi}^{h\{3,1\}} \Big)$ the harmonic embeddings of the form:
\begin{equation*}
\qT{\Phi}^{d\{3,1\}} = \frac{1}{2} \left(\qT{i}^{(4)}_2 + \qT{i}^{(4)}_3 - \qT{i}^{(4)}_1 \right) \quad;\quad
\qT{\Phi}^{h\{3,1\}} = \frac{1}{2} \qT{i}^{(4)}_1,
\end{equation*}
in which $(\qT{i}^{(4)}_1, \qT{i}^{(4)}_2,\qT{i}^{(4)}_3)$ are the fourth-order elementary isotropic tensors defined by Equation \eqref{eq:TIso4}.
Conversely, for any $\tT{T} \in \TT_{(ij)k}$,  $(\tT{H}, \V{v}^d, \V{v}^h) \in \KK^{3}\times\KK^{1} \times\KK^{1}$ are  defined from $\tT{T}$ as follows:
\begin{center}
\begin{tabular}{|c|c|}
  \hline
 $\KK^{1}$ & $\KK^{3}$  \\ \hline
 $ \V{v}^h=\qT{\Pi}^{h\{1,3\}}\tc\tT{T}$ &\\
 $ \V{v}^d=\qT{\Pi}^{d\{1,3\}}\tc\tT{T}$ &  $\tT{H}= \tT{T}- \qT{\Phi}^{d \{3,1\}} \udot \V{v}^d - \qT{\Phi}^{h \{3,1\}} \udot \V{v}^h$ \\
  \hline
\end{tabular}
\end{center}
with $\Big(\qT{\Pi}^{d\{1,3\}},\qT{\Pi}^{h\{1,3\}} \Big)$ the harmonic projectors of the form:
\begin{equation*}
\qT{\Pi}^{d\{1,3\}}=\frac{1}{2}\left(\qT{i}^{(4)}_2+\qT{i}^{(4)}_1-\qT{i}^{(4)}_3\right) \quad;\quad
\qT{\Pi}^{h\{1,3\}}=\qT{i}^{(4)}_3.
\end{equation*}
\end{thm}
From the embedding operators involved in Theorem \ref{thm:HarEmb} a family of projectors can be deduced. 
\begin{prop}\label{cor:Proj3II}
 Let $\sT{I}^{\sgrd}=\frac{1}{2}\left(\sT{i}^{(6)}_8+\sT{i}^{(6)}_{12}\right)$ be the identity tensor on $\sgrd$. 
The following tensors
\ben
\sT{P}^{(3,1h)}:=2(\qT{\Phi}^{h\{3,1\}}\udot\qT{\Phi}^{h\{1,3\}}),\quad
\sT{P}^{(3,1d)}:=(\qT{\Phi}^{d\{3,1\}}\udot\qT{\Phi}^{d\{1,3\}}),\quad
\sT{P}^{(3,3)}:=\sT{I}^{\sgrd}-\sT{P}^{(3,1h)}-\sT{P}^{(3,1d)},
\een
where $\qT{\Phi}^{h\{1,3\}}$ and $\qT{\Phi}^{d\{1,3\}}$ are the transposes of $\qT{\Phi}^{h\{3,1\}}$ and $\qT{\Phi}^{d\{3,1\}}$, constitute a family  of orthogonal projectors on $\HH^{h(3,1)}$, $\HH^{d(3,1)}$ and $\KK^{3}$, respectively.
\end{prop}
In this Proposition, the spaces $\mathbb{H}^{d(3,1)}$ and $\HH^{h(3,1)}$ are defined as follows. $\mathbb{H}^{d(3,1)}$ is the subspace of $\KK^{2}\otimes\KK^{1}$ of tensors orthogonal to elements in $\KK^{3}$:
\begin{equation}
\mathbb{H}^{d(3,1)}: = \left \{\tT{T} \in \KK^{2}\otimes\KK^{1} \mid \forall \tT{H}\in\KK^{3}, \tT{H}\tc\tT{T} = 0 \right \},\ \quad \dim (\mathbb{H}^{d(3,1)})= 2. 
\end{equation}
Unlike the tensors belonging to $\KK^{3}$, the elements of $\mathbb{H}^{d(3,1)}$ are traceless only with respect to their first two indices. 
The space $\HH^{h(3,1)}$  is defined as
\begin{equation}
\mathbb{H}^{h(3,1)}: = \left \{\tT{T} \in \sgrd \mid \forall \tT{H}\in\sgrd, \tT{H}\tc\tT{T} = 0 \right \},\ \quad \dim (\mathbb{H}^{h(3,1)})= 2. 
\end{equation}
It is the subspace of $\sgrd$ of tensors having none vanishing trace $(12)$. 
\begin{prop}\label{prop:CGDec5II}
The tensor $\cT{M}\in\Elac$ admits the uniquely defined \red{Intermediate Block Decomposition} associated to the family of projectors $(\sT{P}^{(3,3)},\sT{P}^{(3,1d)},\sT{P}^{(3,1h)},\qT{P}^{(2,2)},\qT{P}^{(2,0)})$:
\ben
\cT{M}=\cT{M}^{2,3}+\tT{m}^{2,1d}\udot\qT{\Phi}^{d\{1,3\}}+2\tT{m}^{2,1h}\udot\qT{\Phi}^{h\{1,3\}}+2\dT{\Phi}^{\{2,0\}}\otimes\tT{m}^{0,3}+2\left(\dT{\Phi}^{\{2,0\}}\otimes\V{\mu}^{0,1d}\right)\udot\qT{\Phi}^{d\{1,3\}}+4\left(\dT{\Phi}^{\{2,0\}}\otimes\V{\mu}^{0,1h}\right)\udot\qT{\Phi}^{h\{1,3\}} 
\een
in which $\cT{M}^{2,3}\in\KK^2\otimes\KK^3$, $(\tT{m}^{2,1d}, \tT{m}^{2,1h})\in (\KK^2\otimes\KK^1)^2$, $\tT{m}^{0,3}\in\KK^3$,  $(\V{\mu}^{0,1d},\V{\mu}^{0,1h})\in(\KK^1)^2$,  and $\qT{\Phi}^{d\{3,1\}}, \qT{\Phi}^{h\{3,1\}}$ and $\dT{\Phi}^{\{0,2\}}$ are defined, respectively, in Propositions \ref{thm:DecT3} and \ref{thm:DecT2}. 
Those elements are defined from $\cT{M}$ as follows:
\begin{center}
\begin{tabular}{|c|c|c|}
  \hline
 $\TT^{1}$ & $\TT^{3}$ & $\TT^{5}$  \\
  \hline
  $\V{\mu}^{0,1d}:=\dT{\Phi}^{\{0,2\}}:\cT{M}\tc\qT{\Phi}^{d\{3,1\}}$&$\tT{m}^{2,1d}:=\qT{P}^{(2,2)}:\cT{M}\tc\qT{\Phi}^{d\{3,1\}}$&$\cT{M}^{2,3}:=\qT{P}^{(2,2)}:\cT{M}\tc\sT{P}^{(3,3)}$\\
  $\V{\mu}^{0,1h}:=\dT{\Phi}^{\{0,2\}}:\cT{M}\tc\qT{\Phi}^{h\{3,1\}}$  &$ \tT{m}^{2,1h}:=\qT{P}^{(2,2)}:\cT{M}\tc\qT{\Phi}^{h\{3,1\}}$&  \\ 
    &$\tT{m}^{0,3}:=\dT{\Phi}^{\{0,2\}}:\cT{M}\tc\sT{P}^{(3,3)}$ & \\          \hline
\end{tabular}
\end{center}
\end{prop}
\begin{prop}[\textbf{Clebsch-Gordan Harmonic Decomposition of} $\cT{M}\in\Elac$]
The tensor $\cT{M}\in\Elac$ admits the uniquely defined Clebsch-Gordan  Harmonic Decomposition associated to the family of projectors $(\sT{P}^{(3,3)},\sT{P}^{(3,1d)},\sT{P}^{(3,1d)},\qT{P}^{(2,2)},\qT{P}^{(2,0)})$:
\ban
\cT{M}&=&\cT{H}^{2,3}+\tT{H}^{2,1d}\udot\qT{\Phi}^{d\{1,3\}}+2\tT{H}^{2,1h}\udot\qT{\Phi}^{h\{1,3\}}+2\dT{\Phi}^{\{2,0\}}\otimes\tT{H}^{0,3}+\sT{\Phi}^{\{5,1\}}\udot\V{v}^{2,3}\\
&+&(\qT{\Phi}^{\{3,1\}}\udot\V{v}^{2,1d})\udot\qT{\Phi}^{d\{1,3\}}+2(\qT{\Phi}^{\{3,1\}}\udot\V{v}^{2,1h})\udot\qT{\Phi}^{h\{1,3\}}+
2\left(\dT{\Phi}^{\{2,0\}}\otimes\V{\mu}^{0,1d}\right)\udot\qT{\Phi}^{d\{1,3\}}\\
&+&4\left(\dT{\Phi}^{\{2,0\}}\otimes\V{\mu}^{0,1h}\right)\udot\qT{\Phi}^{h\{1,3\}}  
\ean
in which $\cT{H}^{2,3}\in\KK^5$, $(\tT{H}^{2,1d}, \tT{H}^{2,1h},\tT{H}^{0,3})\in(\KK^3)^3$, $(\V{v}^{2,3},\V{v}^{2,1d},\V{v}^{2,1h},\V{v}^{0,1d},\V{v}^{0,1h})\in(\KK^1)^5$.
Those elements are defined from $\cT{M}$ as follows:
\begin{center}
\begin{tabular}{|c|c|c|}
  \hline
 $\KK^{1}$ & $\KK^{3}$ & $\KK^{5}$  \\
  \hline
  $\V{v}^{0,1h} =\V{\mu}^{0,1h}$  & &  \\ 
  $\V{v}^{0,1d} =\V{\mu}^{0,1d}$ & & \\
  $\V{v}^{2,1h} =  \tT{m}^{2,1h}:\id$ &$\tT{H}^{2,1h}=\tT{m}^{2,1h}-\qT{\Phi}^{\{3,1\}}\udot\V{v}^{2,1h}$&\\
  $\V{v}^{2,1d}=\tT{m}^{2,1d}:\id$   &$\tT{H}^{2,1d}=\tT{m}^{2,1d}-\qT{\Phi}^{\{3,1\}}\udot\V{v}^{2,1d}$&\\
  &$\tT{H}^{0,3}= \tT{m}^{0,3}$&\\
  $\V{v}^{2,3}=\tr_{12}(\tr_{13}(\cT{M}^{2,3}))$&&$\cT{H}^{2,3}=\cT{M}^{2,3}-\sT{\Phi}^{\{5,1\}}\udot\V{v}^{2,3}$\\                                                                         
  \hline
\end{tabular}
\end{center}
in which the intermediate quantities are defined in Proposition \ref{prop:CGDec5}  and $\qT{\Phi}^{d\{3,1\}}, \qT{\Phi}^{h\{3,1\}}$ and $\dT{\Phi}^{\{0,2\}}$ are defined, respectively, in Propositions \ref{thm:DecT3} and \ref{thm:DecT2}. 
\end{prop}
\begin{prop}\label{prop:CGDec6II}
The tensor $\sT{A}\in\Elas$ admits the uniquely defined \red{Intermediate Block Decomposition} associated to the family of projectors $(\sT{P}^{(3,3)},\sT{P}^{(3,1d)},\sT{P}^{(3,1h)})$:
\begin{eqnarray*}
\sT{A}&=&\sT{A}^{3,3}+\qT{\Phi}^{d\{3,1\}}\udot\dT{a}^{1d,1d}\udot\qT{\Phi}^{d\{1,3\}}+4\qT{\Phi}^{h\{3,1\}}\udot\dT{a}^{1h,1h}\udot\qT{\Phi}^{h\{1,3\}}\\
&+&\left(\qT{a}^{3,1d}\udot\qT{\Phi}^{d\{1,3\}} +\qT{\Phi}^{d\{3,1\}}\udot\qT{a}^{1d,3}\right)+2\left(\qT{a}^{3,1h}\udot\qT{\Phi}^{h\{1,3\}}+\qT{\Phi}^{h\{3,1\}}\udot\qT{a}^{1h,3}\right)\\
&+&2\left(\qT{\Phi}^{d\{3,1\}}\udot\dT{a}^{1d,1h}\udot\qT{\Phi}^{h\{1,3\}}+\qT{\Phi}^{h\{3,1\}}\udot\dT{a}^{1h,1d}\udot\qT{\Phi}^{d\{1,3\}}\right)
\end{eqnarray*}
in which $\sT{A}^{3,3}\in\KK^3\otimes^s\KK^3$, $(\qT{a}^{3,1d}, \qT{a}^{3,1h})\in(\KK^3\otimes\KK^1)^2$, $(\dT{a}^{1d,1d},\dT{a}^{1h,1h})\in(\KK^1\otimes^s\KK^1)^2$ and $\dT{a}^{1d,1h}\in\KK^1\otimes\KK^1$ and $\qT{\Phi}^{d\{3,1\}}, \qT{\Phi}^{h\{3,1\}}$ are defined in Proposition \ref{thm:DecT3II}.
Those elements are defined from $\sT{A}$ as follows:
\begin{center}
\begin{tabular}{|c|c|c|}
  \hline
 $\TT^{2}$ & $\TT^{4}$ & $\TT^{6}$  \\
  \hline
  $\dT{a}^{1d,1d}:=\qT{\Phi}^{d\{1,3\}}\tc\sT{A}\tc\qT{\Phi}^{d\{3,1\}}$&$\qT{a}^{3,1d}:=\sT{P}^{(3,3)}\tc\sT{A}\tc\qT{\Phi}^{d\{3,1\}}$&$\sT{A}^{3,3}:=\sT{P}^{(3,3)}\tc\sT{A}\tc\sT{P}^{(3,3)}$\\
  $\dT{a}^{1d,1h}:=\qT{\Phi}^{d\{1,3\}}\tc\sT{A}\tc\qT{\Phi}^{h\{3,1\}}$  &$ \qT{a}^{3,1h}:=\sT{P}^{(3,3)}\tc\sT{A}\tc\qT{\Phi}^{h\{3,1\}}$&  \\ 
   $\dT{a}^{1h,1h}:=\qT{\Phi}^{h\{1,3\}}\tc\sT{A}\tc\qT{\Phi}^{h\{3,1\}}$ & &\\         \hline
\end{tabular}
\end{center}
\end{prop}
\begin{prop}[\textbf{Clebsch-Gordan Harmonic Decomposition of} $\sT{A}\in\Elas$]
The tensor $\sT{A}\in\Elas$ admits the uniquely defined  Clebsch-Gordan Harmonic Decomposition associated to the family of projectors $(\sT{P}^{(3,3)},\sT{P}^{(3,1d)},\sT{P}^{(3,1h)})$:
\ban
\sT{A}&=&\sT{H}^{3,3}+\left(\qT{H}^{3,1d}\udot\qT{\Phi}^{d\{1,3\}} +\qT{\Phi}^{d\{3,1\}}\udot\qT{H}^{3,1d}\right)+2\left(\qT{H}^{3,1r}\udot\qT{\Phi}^{h\{1,3\}}+\qT{\Phi}^{h\{3,1\}}\udot\qT{H}^{3,1r}\right)\\
&+&\qT{\Phi}^{d\{3,1\}}\udot\dT{h}^{1d,1d}\udot\qT{\Phi}^{d\{1,3\}}+4\qT{\Phi}^{h\{3,1\}}\udot\dT{h}^{1r,1r}\udot\qT{\Phi}^{h\{1,3\}}\\
&+&\left(\left(\sT{\Phi}^{\{4,2\}}:\dT{h}^{3,1d} \right)\udot\qT{\Phi}^{d\{1,3\}}+\qT{\Phi}^{d\{3,1\}}\udot\left(\sT{\Phi}^{\{4,2\}}:\dT{h}^{3,1d}\right )^{T_{3,1}}\right)\\
&+&2\left(\left(\sT{\Phi}^{\{4,2\}}:\dT{h}^{3,1r}\right) \udot\qT{\Phi}^{h\{1,3\}}+\qT{\Phi}^{h\{3,1\}}\udot\left(\sT{\Phi}^{\{4,2\}}:\dT{h}^{3,1r} \right)^{T_{3,1}}\right)\\
&+&2\left(\qT{\Phi}^{d\{3,1\}}\udot\dT{h}^{1d,1r}\udot\qT{\Phi}^{h\{1,3\}}+\qT{\Phi}^{h\{3,1\}}\udot\dT{h}^{1r,1d}\udot\qT{\Phi}^{d\{1,3\}}\right)\\
&+&\frac{\alpha^{3,3}}{2}\sT{P}^{(3,3)}+\frac{1}{2}\alpha^{1d,1d}\sT{P}^{(3,1d)}+\alpha^{1r,1r}\sT{P}^{(3,1r)}+\alpha^{1d,1r}\left(\qT{\Phi}^{d\{3,1\}}\udot\qT{\Phi}^{h\{1,3\}}+\qT{\Phi}^{h\{3,1\}}\udot\qT{\Phi}^{d\{1,3\}}\right)\\
&+&\beta^{1d,1r}\left(\qT{\Phi}^{d\{3,1\}}\udot\dT{\epsilon}\udot\qT{\Phi}^{h\{1,3\}}-\qT{\Phi}^{h\{3,1\}}\udot\dT{\epsilon}\udot\qT{\Phi}^{d\{1,3\}}\right),
\ean
in which $\sT{H}^{3,3}\in\KK^6$, $(\qT{H}^{3,1d}, \qT{H}^{3,1r})\in(\KK^4)^2$, $(\dT{h}^{3,1d},\dT{h}^{3,1r},\dT{h}^{1d,1r},\dT{h}^{1d,1d}\dT{h}^{1r,1r})\in(\KK^2)^5$, $(\alpha^{3,3},\alpha^{1d,1d},\alpha^{1r,1r},\alpha^{1r,1d})\in(\KK^0)^4$ and $\beta^{1r,1d}\in\KK^{-1}$.
Those elements are defined from $\sT{A}$ as follows:
\begin{center}
\begin{footnotesize}
\begin{tabular}{|c|c|c|c|c|}
  \hline
 $\KK^{-1}$ & $\KK^{0}$ & $\KK^{2}$ &  $\KK^{4}$ & $\KK^{6}$ \\
  \hline
  $\beta^{1d,1h} =\dT{a}^{1d,1h}:\Jd$  & $\alpha^{1d,1h} =\dT{a}^{1d,1h}:\id$ & $\dT{h}^{1d,1h}=\dT{a}^{1d,1h}:\qT{P}^{(2,2)}$ &&\\
                                                            & $\alpha^{1h,1h} =\dT{a}^{1h,1h}:\id$ & $\dT{h}^{1h,1h}=\dT{a}^{1h,1h}:\qT{P}^{(2,2)}$ &&\\
                                                            & $\alpha^{1d,1d} =\dT{a}^{1d,1d}:\id$ & $\dT{h}^{1d,1d}=\dT{a}^{1d,1d}:\qT{P}^{(2,2)}$ &&\\
                                                            &         &$\dT{h}^{3,1h}=\tr_{14}(\qT{a}^{3,1h})$&$\qT{H}^{3,1h}=\qT{a}^{3,1h}-\sT{\Phi}^{\{4,2\}}:\dT{h}^{3,1h}$&\\
           &                                                          &$\dT{h}^{3,1d}=\tr_{14}(\qT{a}^{3,1d})$&$\qT{H}^{3,1d}=\qT{a}^{3,1d}-\sT{\Phi}^{\{4,2\}}:\dT{h}^{3,1d}$&\\
      &$\alpha^{3,3}=\sT{A}^{3,3}\rdots{6}\sT{P}^{(3,3)}$&&&    $\sT{H}^{3,3}=\sT{A}^{3,3}-\frac{\alpha^{3,3}}{2}\sT{P}^{(3,3)}$\\ 
  \hline
\end{tabular}
\end{footnotesize}
\end{center}
in which the intermediate quantities are defined in Proposition \ref{prop:CGDec6II}. 
\end{prop}
}

\newpage
 \bibliographystyle{abbrv}
\bibliography{biblio}

\end{document}